\documentclass{article}

\usepackage{amsmath}
\usepackage{booktabs}   
\usepackage{subcaption} 
\usepackage[colorinlistoftodos]{todonotes}
\usepackage{mathrsfs}
\usepackage{authblk}
\usepackage[numbers]{natbib}
\usepackage{amssymb, amsthm}
\usepackage{hyperref}

\usetikzlibrary{backgrounds, calc}

\newtheorem{theorem}{Theorem}[section]
\newtheorem{lemma}[theorem]{Lemma}
\newtheorem{corollary}[theorem]{Corollary}
\newtheorem{definition}[theorem]{Definition}
\newtheorem{proposition}[theorem]{Proposition}

\newcommand{\shuf}{\mbox{$\ \sqcup\hspace*{-2.33ex}\perp\ $}}

\DeclareMathOperator{\var}{var}

\newcommand{\ta}{\mathtt{a}}
\newcommand{\tb}{\mathtt{b}}
\newcommand{\tc}{\mathtt{c}}

\newcommand{\FO}{\mathsf{FO}}

\DeclareMathOperator{\Shuffle}{Shuffle}
\DeclareMathOperator{\Insert}{Insert}
\DeclareMathOperator{\Erase}{Erase}
\DeclareMathOperator{\Morphism}{MorphIm}
\DeclareMathOperator{\Projection}{Projection}
\DeclareMathOperator{\Abelian}{AbelianEq}
\DeclareMathOperator{\Multiply}{Multiply}
\DeclareMathOperator{\Subword}{Subword}
\DeclareMathOperator{\Length}{Length}

\begin{document}

\title{The Satisfiability of Extended Word Equations: The Boundary Between Decidability and Undecidability}         

\author[1]{Joel Day}

\author[2]{Vijay Ganesh}

\author[2]{Paul He}

\author[1]{Florin Manea}

\author[1]{Dirk Nowotka}

\affil[1]{\small Kiel University}
\affil[2]{\small University of Waterloo}

\date{}

\maketitle
\begin{abstract} 
The study of word equations (or the existential theory of equations over free monoids) is a central topic in mathematics and theoretical computer science. The problem of deciding whether a given word equation has a solution was shown to be decidable by Makanin in the late 1970s, and since then considerable work has been done on this topic. In recent years, this decidability question has gained critical importance in the context of string SMT solvers for security analysis. Further, many extensions (e.g., quantifier-free word equations with linear arithmetic over the length function) and fragments (e.g., restrictions on the number of variables) of this theory are important from a theoretical point of view, as well as for program analysis applications. Motivated by these considerations, we prove several new results and thus shed light on the boundary between decidability and undecidability for many fragments and extensions of the first order theory of word equations.
\end{abstract}

\section{Introduction}
\label{sec:intro}

A {\em word equation} is a formal equality $U = V$, where $U$ and $V$ are words (called the left, respectively, right side of the equation) over an alphabet $A \cup X$; $A=\{\ta,\tb,\tc,\ldots\}$ is the alphabet of \emph{constants} or {\em terminals} and $X = \{x_1, x_2, x_3, \ldots\}$ is the alphabet set of \emph{variables}. A \emph{solution} to the equation $U=V$ is a morphism $h : (A \cup X)^* \to A^*$ that acts as the identity on $A$ and satisfies $h(U) = h(V)$; $h$ is called the assignment to the variables of the equation. For instance, $U =x_1 \ta \tb x_2$ and $V = \ta x_1 x_2 \tb$ define the equation $x_1 \ta \tb x_2 = \ta x_1 x_2 \tb$,
whose solutions are the morphisms $h$ with $h(x_1) = \ta^k$, for $k\geq 0$, and $h(x_2) = \tb^\ell$, for $\ell \geq 0$. An equation is {\em satisfiable} (in $A^*$) if it admits a solution $h: (A \cup X)^* \to A^*$. A set (or system) of equations is satisfiable if there exists an assignment of the variables of the equations in this set that is a solution for all equations. In logical terms, word equations are often investigated as fragments of the first order theory $\FO(A^*,\cdot)$ of strings. \citet{karhumaki2000} showed that deciding the satisfiability of a system of word equations, 
that is, checking the truth of formulas from the existential theory $\Sigma_1$ of $\FO(A^*,\cdot)$,
can be reduced to deciding the satisfiability of a (more complex) single word equation that encodes the respective system. 

The existential theory of word equations (simply called theory of word equations, if not mentioned otherwise) has been studied for decades in
mathematics and theoretical computer science with a particular focus on the decidability of the satisfiability of logical formulae defined over word equations. 
In 1946, \citet{quine1946} proved that the first-order 
theory of word equations is equivalent to the first-order theory of arithmetic, which is known to be
undecidable. In order to solve Hilbert's tenth problem \cite{hilbert1900} in the negative,
Markov showed a reduction from word equations to Diophantine equations (see \cite{lothaire,makanin1977} and the references therein), in the~hopes that word equations would prove to be undecidable. 
However, \citet{makanin1977} proved in 1977 that the satisfiability of word equations \emph{is}
decidable. Though Markov's approach was unsuccessful, a related idea 
can be tried again based on extended theories of word equations. \citet{matiyasevich1968} showed in 1968  
a reduction from the more powerful theory~of word equations with linear length constraints  (i.e., linear relations between word lengths) to Diophantine equations. Whether this theory is decidable remains a major open~problem.

After Makanin showed that the satisfiability of word equations is decidable, the focus shifted towards identifying the complexity of 
deciding the satisfiability of an equation. After a series of intermediate results \cite{lothaire}, 
\citet{plandowski1999} showed that this problem is in PSPACE. In a series of recent papers \cite{jez2013,jez2017}, Je\.z
applied a new technique called recompression to word equations to first simplify the existing 
proof that the satisfiability of word equations can be decided in polynomial space, and then to show that
this can actually be decided in linear space. However, there is a mismatch between the aforementioned upper bounds and the only known lower bound: solving word equations is NP-hard. 

In recent years, deciding the satisfiability of systems of word equations has also become an important problem in fields such as formal
verification and security where string solvers such as HAMPI \cite{kiezun2009}, CVC4 \cite{barrett2011}, Stranger \cite{yu2010}, ABC~\cite{aydin2015}, Norn \cite{abdulla2015}, S3P \cite{trinh2016} and Z3str3 \cite{berzish2017}
have become more popular. However, in practice more functionality than just word equations is required in many cases, so solvers often extend 
the theory of word equations with certain functions (e.g., linear arithmetic over the length, replace-all, extract, reverse, etc.) and predicates (e.g., numeric-string conversion predicate, regular-expression membership, etc.). Due to the complexity of solving word 
equations and undecidability of many of these extensions, none of these solvers have a complete algorithm. To this end, for example, the extension of word equations with a $replace-all$ operator was shown to be undecidable in \cite{lin2016}.

In \cite{karhumaki2000} the authors introduce the notion of languages expressible by word equations as, intuitively, the set of solutions that an equation may have. 
It is immediate that the satisfiability problem for systems of word equations whose variables are constrained by expressible languages is decidable. However, in many extensions that 
are used in conjunction with classical word equations (both in practical and theoretical settings) the constraints are not expressible by word equations. 
To this end, we can mention regular (or rational) constraints, constraints based on involutions (such as the mirror image), or length constraints, none of which are expressible \cite{buchi1990,karhumaki2000}.
As mentioned above, whether the theory of word equations enhanced with a length function is decidable is still a major open problem. But on the other hand, the satisfiability of word equations with regular constraints \cite{lothaire} or with involutions \cite{diekert2016} is decidable in both cases.

In this setting, our work aims to provide a better understanding of the boundary between extensions of the theory of word equations for which satisfiability is decidable or, respectively, undecidable.

\paragraph{Our Contributions:} On the one hand, we show that for a series of natural and practically interesting extensions of word equations, the satisfiability problem is undecidable. On the other hand, we address the decidability of the theory of word equations with length constraints, and show for some classes of word equations with restricted forms and length constraints the satisfiability problem is decidable. We also prove several expressibility results that shed light on the relative power of word equations vis-a-vis other kinds of formal language representations such as regular expressions and context-free grammars.

Our first result is related to expressibility. As noted before, many simple constraints are not expressible by systems of satisfiable word equations (sat-equations), but can be easily expressed by requiring that some equations are unsatisfiable (unsat-equations). For instance, if one wants to define the set of words of the form $\ta w \tb w \tc$ where $w$ is a string that contains no symbol $\tc$, this can be specified by requiring the equation $x_1=\ta x_2 \tb x_2 \tc$ to be satisfiable and the equation $x_2=x_3\tc x_4$ to be unsatisfiable, i.e., not true for any assignment of the variables $x_3$ and $x_4$. It is an easy exercise to show, using the techniques in \cite{karhumaki2000}, that $\{\ta w \tb w \tc\mid w$ contains no symbol $\tc\}$ is not expressible by word equations. 
We are interested whether the satisfiability of systems of sat- and unsat-equations is decidable. In this setting, one is given two sets of equations that may share variables: the set of sat-equations and the set of unsat-equations; both sets might also contain negated equations. One has to decide whether there exists an assignment of the variables occurring in the sat-equations that satisfies this entire set, such that no matter what way we assign the rest of the variables at least one of the unsat-equations is not satisfied. We show that this gives an alternative characterization of the the $\Sigma_2$ fragment of $\FO(A^*,\cdot)$, i.e., the fragment of $\exists\forall$ quantified first order formulae over word equations. Thus, the satisfiability of such systems is undecidable. To obtain these results, we show that deciding the truth of $\Sigma_2$ formulae is equivalent to deciding the truth of a formula consisting of a single $\exists\forall$-quantified negated equation, which, at its turn, can be encoded as the satisfiability of a system of sat- and unsat-equations. As the Inclusion of Pattern Languages problem (see~\cite{dominik,jiang}) can be encoded as such a system as well, it follows that $\Sigma_2$ is undecidable even when the alphabet of terminals is of size~$2$. This result is complemented by the fact that deciding the truth of formulae from the \emph{positive} $\Sigma_2$ fragment of $\FO(A^*,\cdot)$ (i.e., $\exists\forall$ quantified formulae obtained by iteratively applying only conjunction and disjunction to word equations of the form $U=V$) is decidable. This series of observations is strongly related to the work of \cite{quine1946,Vijay_HVC,Durnev}, in which it was shown that the validity of sentences from the positive $\Pi_2$ fragment of $\FO(A^*,\cdot)$ (i.e., the quantifier alternation was, in that case, $\forall\exists$) is undecidable, as well as to the results of \cite{russians} in which it was shown that the truth of arbitrarily quantified positive formulae over word equations is decidable \emph{over an infinite alphabet of terminals}. Note that our positive result does not contradict those in the aforementioned papers. Indeed, when trying to check whether a $\Pi_2$ formula over $A^*$ is valid, one can reduce this to checking whether a $\Sigma_2$ formula over arbitrary word equations is true over $A^*$. However, the resulting formula may contain negated equations (that is, $U\neq V$ atoms), so it would not have the required form in our decidability result. In fact, as soon as we allow universally quantified negated equations in the $\Sigma_2$ formulae, we obtain an undecidable fragment. Also, our positive result does not follow from \cite{russians}, where the requirement that the alphabet of terminals is infinite was crucial.

Our second line of results presents a series of undecidability results for
the $\Sigma_1$ fragment of $\FO(A^*,\cdot)$, the theory of word equations,
extended with simple predicates or functions. We show that adding to word equations either length constraints and a function that maps a string to its integer value, or, alternatively, just constraints imposing that two strings have the same number of occurrences of two fixed letters, leads to an undecidable theory. Also, the same holds when we extend the theory of word equations with constraints requesting that two word are abelian equivalent (they have the same Parikh vector), or with constraints imposing that a string is the morphic image of another one, etc. These results are related to the study of theories of quantifier free word equations constrained by very simple relations, see~\cite{buchi1990} for instance. While our results do not settle the decidability of the theory of word equations with length constraints, they give the intuitive idea that the theory of word equations enhanced with predicates providing very little control on the combinatorial structure of the solutions of the equation (and not necessarily with any control on the length) becomes undecidable. 

We also show the following positive results. Firstly, the satisfiability of quantifier free positive formulae over word equations with linear length constraints, in which we have only one terminal (occurring zero, one, or multiple times) and no restriction on the usage of variables, is decidable, and, moreover, NP-complete, no matter the alphabet over which we search for the solutions; the decidability is preserved when considering positive $\Sigma_2$ formulae of this kind. To this end, we also show that if we allow negated equations in our quantifier free formulae (so arbitrary $\Sigma_1$ formulae), we obtain a theory that is decidable if and only if the general theory of equations with length constraints is decidable. Thus, the study of equations with only one terminal seems motivated to us, despite their simple structure. Secondly, the satisfiability of quantifier free equations with linear length constraints, which have a strictly regular-ordered form (each variable occurs exactly once in each side, and the order in which the variables occur is the same) is decidable, even when we add regular constraints. We also show that, in the latter case, if the regular constraints are given by DFAs, the satisfiability problem is NP-complete.

The first positive result mentioned above is connected to the study of constant-free word equations, for which the existence of parametrisable solutions was thoroughly investigated (see, e.g., \cite{Budkina,DirkTCS2003,SaarelaDLT}). We show that equations with a single terminal symbol (occurring several times) always admit a certain type of structurally simple solutions, which allows for a reformulation of their satisfiability problem into an integer linear programming problem, so both the decidability and complexity results follows. The second result is related to the investigations initiated in \cite{DLT2016,MFCS2017}, in which the authors were interested in the complexity of solving equations of restricted form. In the most significant result of \cite{MFCS2017}, it was shown that deciding the satisfiability of strictly regular-ordered equations (with or without regular constraints) is NP-complete, which makes this class of word equations one of simplest known classes of word equations that are hard to solve. It seems interesting to us that the NP-completeness of the satisfiability problem is preserved for regular-ordered equations with linear length constraints. However, our proof does not seem to scale to less restricted classes of equations. 

\paragraph{Organization:}
The organization of the paper is as follows. In Section~\ref{sec:preliminaries}
we introduce the basic notions involved in the problem of solving word equations. In Section~\ref{sec:unsat} we present the series of observations regarding systems of sat- and unsat-equations, as well as those related to the fragment $\Sigma_2$ of the first order theory of word equations. In Section~\ref{sec:undec} we present a series of undecidable extensions of the quantifier-free theory of word equations, while in Section~\ref{sec:dec} we present the decidable cases. We conclude with Section~\ref{sec:hierarchy}, where we present a map of the results of this paper, emphasizing the steps we took towards delineating the boundary between decidability and undecidability in this context.

\section{Preliminaries}
\label{sec:preliminaries}
Let $A$ be an alphabet of letters (or symbols). We denote by $A^*$ the set of all words over $A$; by $\varepsilon$ we denote the empty word. Note that $A^*$ is a monoid w.r.t. the concatenation of words. Let $|w|$ denote the length of a word $w$. 
For $1\leq i\leq j\leq |w|$ we denote by $w[i]$ the letter on the $i$\textsuperscript{th} position of $w$. A word $w$ is $p$-periodic for $p\in \mathbb{N}$ ($p$ is called a period of~$w$) if $w[i]=w[i+p]$ for all $1\leq i\leq |w|-p$; the smallest period of a word is called its period. 
Let $w=v_1v_2v_3$ for some words $v_1,v_2,v_3\in A^*$, then $v_1$ is called prefix of $w$, $v_1,v_2,v_3$ are factors of $w$, and $v_3$ is a suffix of $w$. Two words $w$ and $u$ are called conjugate if there exist non-empty words $v_1,v_2$ such that $w=v_1v_2$ and $u=v_2v_1$. A word $v\in A^*$ is a subword of $w\in A^*$ if $v=v_1\ldots v_k$, with $v_i\in A^*$, and $w=u_0v_1u_1\cdots v_ku_k$, with $u_i\in A^*$. A word $z\in A^*$ is in the shuffle of $u,v\in A^*$, denoted $z\in x\shuf y$, if $z=u_1v_1\cdots u_kv_k,$ with $u_i,v_i\in A^*$, and $u=u_1\cdots u_k$, $v=v_1\cdots v_k$. Two words $u,v\in A^*$ are abelian equivalent if $|u|_a=|v|_a$, for all $a\in A$.

The following lemma is well known (see, e.g., \cite{lothaire}). 
\begin{lemma}[Commutativity Equation]\label{lem:commute}
Let $v_1, v_2 \in A^*$. Then $v_1v_2 = v_2v_1$ if and only if there exists $w \in A^*$ and $p,q \in \mathbb{N}_0$ such that $v_1 = w^p$ and $v_2 = w^q$.
\end{lemma}

Let $A=\{\ta,\tb,\tc,\ldots\}$ be a finite alphabet of \emph{constants} and let $X = \{x_1, x_2, x_3, \ldots\}$ be an alphabet of \emph{variables}. Note that we assume $X$ and $A$ are disjoint, and unless stated otherwise, that $|A|\geq 2$. A word $\alpha \in (A \cup X)^*$ is usually called a {\em pattern}. For a pattern $\alpha$ and a letter $z \in A \cup X$, let $|\alpha|_z$ denote the number of occurrences of $z$ in~$\alpha$; $\var(\alpha)$ denotes the set of variables from $X$ occurring in $\alpha$. A morphism $h : (A \cup X)^* \to A^*$ with $h(a) = a$ for every $a \in A$ is called a \emph{substitution}. We say that  $\alpha \in (A \cup X)^*$ is \emph{regular} if, for every $x \in \var(\alpha)$, we have $|\alpha|_{x} = 1$; e.\,g., $\ta x_1 \ta x_2 \tc x_3 x_4 \tb$ is regular. Note that $L(\alpha) = \{h(\alpha) \mid h \text{ is a substitution}\}$ (the pattern language of $\alpha$) is regular when $\alpha$ is regular. 

A (positive) \emph{word equation} is a tuple $(U,V) \in (A \cup X)^* \times (A \cup X)^*$; we usually denote such an equation by $U = V$, where $U$ is the left hand side (LHS, for short) and $V$ the right hand side (RHS) of the equation. A negative word equation is the negation of a word equation, i.e., $\lnot (U=v)$ or $U\neq V$.

A \emph{solution} to an equation $U= V$ (respectively, $U\neq V$), over an alphabet $A$, is a substitution $h$ mapping the variables of $UV$ to words from $A^*$ such that $h(U) = h(V)$ (respectively, $h(U)\neq h(V)$). $h(U)$ is called the \emph{solution word} and the length of a solution $h$ of the equation $U=V$ is $|h(U)|$. A solution of shortest length to an equation is called minimal. Note that we might ask whether a positive or negative equation has a solution over an alphabet larger than the alphabet of terminals that actually occur in the respective equation. A word equation is \emph{satisfiable} over $A$ if it has a solution over $A$, and the \emph{satisfiability problem} is to decide for a given word equation whether or not it is satisfiable in some given alphabet $A$. 

We briefly recall the results of \citet{karhumaki2000}. In~\cite{karhumaki2000} it is shown that, given two equations $E$ and $E'$, one can construct the equations $E_1$, $E_2$, and $E_3$ that are satisfiable if and only if $E\land E'$, $E\lor E'$, $\lnot E$ are, respectively, satisfiable. In this construction, $E_1$ contains exactly the variables of $E$ and $E'$, while in $E_2$ and $E_3$ new variables are added with respect to those in the given equations. We use this result to show that for every quantifier free first order formula over word equations we can construct a single equation that may contain extra variables and terminals, and is satisfiable if and only if the initial formula was satisfiable. Moreover, the values the variables of the initial equations may take in the satisfying assignments of the new equation are exactly the same values they took in the satisfying assignments of the initial formula. We also use in several occasions the following result from~\cite{karhumaki2000}.
 \begin{lemma}\label{lem:conjunctions}
Let $U,V,U^\prime,V^\prime \in (X\cup A)^*$. Let $Z_1 = U \ta U^\prime U \tb U^\prime $ and $Z_2 = V \ta V^\prime V \tb V^\prime$. Then for any substitution $h : X^* \to A^*$, $h(Z_1) = h(Z_2)$ if and only if $h(U) = h(V)$ and $h(U^\prime) = h(V^\prime)$. 
\end{lemma}

In Section \ref{sec:dec} of this paper we address equations with restricted form. A word equation $U=V$ is regular if both $U$ and $V$ are regular patterns. We call a regular equation {\em ordered} if the order in which the variables occur in both sides of the equation is the same; that is, if $x$ and $y$ are variables occurring both in $U$ and $V$, then $x$ occurs before $y$ in $U$ if and only if $x$ occurs before $y$ in $V$. Moreover, we say a regular-ordered equation is {\em strict} if each variable occurs in both sides. For instance $x_1\ta x_2x_3 \tb =x_1\ta x_2 \tb x_3$ is strictly regular-ordered while $x_1 \ta = x_1 x_2$ is regular-ordered (but not strictly since $x_2$ occurs only on one side) and $x_1\ta x_3x_2 \tb =x_1\ta x_2 \tb x_3$ is regular but not regular-ordered.

The results of the last section also consider equations with regular constraints and linear length constraints defined as follows. Given a word equation $U = V$, a set of {\em linear length constraints} is a system $\theta$ of linear Diophantine equations where the unknowns correspond to the lengths of possible substitutions of each variable $x \in X$. Moreover, given a variable $x \in X$, a {\em regular constraint} is a regular language $L_x$ given by a finite automata. The satisfiability of word equations with linear length and/or regular constraints is the question of whether a solution $h$ exists satisfying the system $\theta$ and/or such that $h(x) \in L_x$ for each $x \in X$. 



\section{Systems of Sat- and Unsat-Equations}
\label{sec:unsat}
\newcommand{\mS}{{\mathcal S}}
\newcommand{\mU}{{\mathcal U}}

We begin by introducing the main concept of this section. Let us assume for the rest of this section that we only work with equations over an alphabet $A$ with at least $2$ letters.
\begin{definition}
Let $A$ be an alphabet of constants, $|A|\geq 2$, and $X$ and $Y$ two disjoint alphabets of variables. Let $\mS=\{e_1,\ldots,e_n\}$ and $\mU=\{f_1,\ldots,f_m\}$ be two finite sets where each $e_i$ is either $U_i=V_i$ or $\lnot (U_i=V_i)$ for some $U_i,V_i\in (A\cup X)^*$, and $f_i$ is either $U'_i=V'_i$ or $\lnot (U'_i=V'_i)$,
for some $U'_i,V'_i\in (A\cup X\cup Y)^*$. We say that $\mS$ and $\mU$ define a system of sat- and unsat-equations over $A$, denoted $(\mS, \mU)$. 

$(\mS, \mU)$ is satisfiable over $A$ if there exists an assignment of the variables from $X$ to words from $A^*$, that satisfies all $e_i\in \mS$, with $1\leq i\leq n$, and for all assignments of the variables of $Y$ at least one of $f_j$ is not satisfied, for $1\leq j\leq m$. 
\end{definition}

Essentially, the class of systems of sat- and unsat-equations extends the existential theory of word equations by adding the possibility to express some undesirable properties of the solutions of these equations via unsat-equations (as exemplified in Section \ref{sec:intro}). Our first results show that deciding the satisfiability of systems of sat- and unsat-equations over $A$ is equivalent to deciding the truth of some very simple $\Sigma_2$ formulae in $A^*$.  
\begin{lemma}\label{lem:sys1}
Let $(\mS,\mU)$ be a system of sat- and unsat-equations over $A$, with $X=\{x_1,\ldots,x_n\}$ the variables occurring in $\mS$ and $Y=\{y_1,\ldots,y_m\}$ be the variables occurring only in $\mU$. Then there exists a $\Sigma_2$ formula 
\begin{align*}\phi=\exists x_1,\ldots,x_t. \forall y_1,\ldots,y_m. e \land (f_1\lor \ldots \lor f_p), 
\end{align*}
where $e$ is a positive equation with variables from $\{x_1,\ldots,x_t\}$, with $t\geq n$, and $f_1,\ldots ,f_p$ are (positive and negative) equations with variables from $\{x_1,\ldots,x_n,y_1,\ldots,y_m\}$, such that $\phi'$
holds in $A^*$ if and only if $(\mS,\mU)$ is satisfiable. 
\end{lemma}
\begin{proof}
Assume $\mS=\{e_1,\ldots,e_a\}$ and $\mU=\{f_1,\ldots,f_b\}$. Let 
\begin{align*}\phi=\exists x_1,\ldots,x_n.\forall y_1,\ldots y_m. e_1\land \ldots \land e_a\land (\lnot f_1 \lor \ldots \lnot f_b)
\end{align*}
It is immediate that $\phi$ is true in $A^*$ if and only if $(\mS,\mU)$ is satisfiable. According to \cite{karhumaki2000}, we can reduce $\phi$ to a formula 
\begin{align*}\phi'=\exists x_1,\ldots,x_n,x_{n+1},& \ldots,x_{t}.\forall y_1,\ldots y_m. \\
& e\land (\lnot f_1 \lor \ldots \lor \lnot f_b)
\end{align*}
where $e$ is a single word equation $U=V$, $U,V\in (A\cup X')^*$ with $ X'=\{x_1,\ldots,x_t\}$. If all the equations $e_i$ are positive then $t=n$ (so no new variables are added), while if at least one of $e_i$ is negative then $t>n$. 
 \end{proof}
We can also prove the following converse result.
\begin{lemma}\label{lem:sys2}
Let $\phi$ be a $\Sigma_2$ formula
\begin{align*}
\phi=\exists x_1,\ldots,x_n. \forall y_1,\ldots,y_m. e \land (f_1\lor \ldots \lor f_p),\end{align*}
where $e$ is a positive equation with variables from $\{x_1,\ldots,x_n\}$, and $f_i$ is either $U_i=V_i$ or $\lnot (U_i=V_i)$, for some $U_i,V_i\in (A\cup \{x_1,\ldots,x_n,y_1,\ldots,y_m\})^*$. Then there exists a system $(\mS,\mU)$ over $A$ that is satisfiable if and only if $\phi$ holds in $A^*$. 
\end{lemma}
\begin{proof}Let $(\mS,\mU)$ be defined as follows:
\begin{align*}\mS &= \{e\} \cup\{x_i=x_i\mid 1\leq i\leq n\}\\
\mU &= \{\lnot f_i\mid 1\leq i\leq p \}.
\end{align*}
It is immediate that $(\mS,\mU)$ is satisfiable if and only if $\phi$ holds in $A^*$. Note that we added the equations $x_i=x_i$ in $\mS$ to ensure that their existential quantification is preserved when trying to solve the system.
\end{proof}

Therefore, each system of sat- and unsat- equations is equivalent to a $\Sigma_2$ formula of very restricted form. As a consequence, we get that the \emph{Inclusion of Pattern Languages} problem (IPL, for short) can be encoded by the satisfiability problem for a system of sat- and unsat-equations. In IPL, one is given two patterns $\alpha \in (A\cup X)^*$ and $\beta \in (A\cup Y)^*$, where $A$ is an alphabet of constants with at least two distinct letters and $X$ and $Y$ are disjoint sets of variables, and has to decide whether $L(\alpha)\subseteq L(\beta)$. 

\begin{theorem}\label{IPL}
Deciding IPL for $\alpha\in (A\cup \{x_1,\ldots,x_n\})^*$ and $\beta\in (A\cup \{y_1,\ldots,y_m\})^*$ can be reduced to deciding whether the following formula holds or not in $A^*$:
\[\exists x_1,\ldots,x_n. \forall y_1,\ldots,y_m. \alpha \neq \beta.\]
\end{theorem}
\begin{proof}
Let $\alpha$ and $\beta$ be the input patterns for IPL. Assume $\alpha=w_0x_1w_2\cdots x_nw_n$ and $\beta=v_0y_1v_2\cdots y_mv_m$ with $X=\{x_1,\ldots,x_n\}$ and $Y=\{y_1,\ldots,y_m\}$ sets of variables and $w_i,v_j\in A^*$ for $1\leq i\leq n,1\leq j\leq m$. Then, $L(\alpha)\subseteq L(\beta)$ if and only if the following formula over word equations is true in $A^*$ 
\begin{align*}
\forall x_1,\ldots,x_n. \exists y_1,\ldots,y_m.w_0x_1w_2\cdots x_nw_n =v_0y_1v_2\cdots y_mv_m
\end{align*}

But this formula is true in $A^*$ if and only if the following formula is false in $A^*$:
\begin{align*}
\exists x_1,\ldots,x_n. \forall y_1,\ldots,y_m.w_0x_1w_2\cdots x_nw_n \neq v_0y_1v_2\cdots y_mv_m
\end{align*}

Checking whether this formula is true (or false) is equivalent, according to Lemmas \ref{lem:sys1} and \ref{lem:sys2} to checking whether a system of sat- and unsat-equations is satisfiable. 
\end{proof}

Theorem \ref{IPL} and Lemma \ref{lem:sys2} shows that deciding IPL for the patterns $\alpha$ and $\beta$ is reducible to solving a system of sat- and unsat-equations $(\mS,\mU)$ with $\mS$ containing only trivial equations $x=x$,  for all variables $x$ occurring in $\alpha$, and $\mU$ a single positive equation $\alpha=\beta$. As IPL is undecidable for terminal alphabets of size $2$ or more, this immediately shows that checking the satisfiability of systems of sat- and unsat-equations is undecidable, over alphabets of size at least $2$.

We are now ready to prove that deciding in general the satisfiability of $\Sigma_2$ formulae over $A^*$ is equivalent to checking the truth value of a formula $\exists x_1,\ldots,x_n. \forall y_1,\ldots,y_m. U\neq V$ in $A^*$, with $U,V$ patterns (whose sets of variables are, however, not necessarily disjoint, as in IPL).
\begin{theorem}\label{sigma2collapse}
For every formula  $\phi$ in the $\Sigma_2$ fragment of $FO(A^*,\cdot)$ we can construct a formula 
\[\psi=\exists x_1,\ldots,x_n. \forall y_1,\ldots,y_m. U\neq V, \]
with $U,V\in (A\cup\{x_1,\ldots,x_n,y_1,\ldots,y_m\})^*$, such that $\phi$ holds in $A^*$ if and only if $\psi$ holds in $A^*$.
\end{theorem}
\begin{proof}
The formula $\phi $ holds in $A^*$ if and only if $\neg \phi$ is false in $A^*$. If $\phi = \exists x_1,\ldots,x_n.\forall y_1,\ldots,y_\ell. \phi',$ we get 
\[\neg \phi = \forall x_1,\ldots,x_n.\exists y_1,\ldots,y_\ell. \neg \phi'.\]
Since $\neg \phi '$ is a quantifier free first order formula over word equations, by \cite{karhumaki2000}, we can construct an equation $U=V$ over an extended set of variables ($y_{\ell+1},\ldots,y_m$ are the newly added variables), such that $\neg \phi$ is false in $A^*$ if and only if 
$\forall x_1,\ldots,x_n.\exists y_1,\ldots,y_m. U=V$
is also false in~$A^*$. 

Finally, let $\psi=\neg (\forall x_1,\ldots,x_n.\exists y_1,\ldots,y_m. U=V)$. Thus:
\[\psi = \exists  x_1,\ldots,x_n.\forall y_1,\ldots,y_m. U\neq V.\]
The formula $\psi$ holds in $A^*$ if and only if $\neg \phi$ is false, so if and only if $\phi$ holds in $A^*$.
\end{proof}

The result in Theorem \ref{sigma2collapse} seems somehow surprising to us, as it shows that checking the truth of an arbitrary $\Sigma_2$ formula reduces to checking the truth of a single negative equation ($\exists\forall$-quantified). Note that applying the results of \citet{karhumaki2000} to the initial arbitrary formula would have only lead to an $\exists\forall\exists$ quantified positive equation, so no longer a $\Sigma_2$ formula. We also get, as a consequence, that checking the truth of an arbitrary $\Sigma_2$ formula in $A^*$ is equivalent to solving a system of sat- and unsat-equations over $A^*$. So, the framework we defined here provides an alternative characterisation for the $\Sigma_2$ fragment of $FO(A^*,\cdot)$. 

By Theorem \ref{sigma2collapse} and Lemma \ref{lem:sys2}, solving systems of sat- and unsat-equations $(\mS,\mU)$ with $\mU$ containing at least one positive equation is undecidable. Let us see the status of such systems where $\mU$ contains only negative equations. 

Let $\Sigma^+_2$ be the \emph{positive} $\Sigma_2$ fragment of $\FO(A^*,\cdot)$ (i.e., $\exists\forall$-quantified formulae obtained by iteratively applying conjunction and disjunction to word equations of the form $U=V$). By Lemma \ref{lem:sys1} solving a system $(\mS,\mU)$ with $\mU$ containing only negative equations can be reduced to checking the truth of a $\Sigma^+_2$ equation in $A^*$.  The converse also holds.
\begin{lemma}\label{lem:sys3}
Let $\phi$ be a $\Sigma^+_2$ formula from $FO(A^*,\cdot)$. Then there exists a system $(\mS,\mU)$, with $\mU$ containing only negative equations, that is satisfiable if and only if $\phi$ holds in $A^*$. 
\end{lemma}
\begin{proof}
First we bring $\phi$ to the disjunctive normal form. We obtain a formula $\phi'=\exists x_1,\ldots,x_n.\forall y_1,\ldots y_m. c_1\lor \ldots \lor c_p$, where each $c_i$, with $1\leq i\leq p$, is a conjunction of factors of the form $U=V$ where $U,V$ are patterns over the variable alphabet $\{x_1,\ldots,x_n,y_1,\ldots,y_m\}$. By Lemma \ref{lem:conjunctions} we obtain that $\phi'$ can be further rewritten as a formula $\phi''=\exists x_1,\ldots,x_n.\forall y_1,\ldots y_m. e_1\lor \ldots \lor e_\ell $ where each $e_i$ is a word equation $U_i=V_i$ with variables from $\{x_1,\ldots,x_n,y_1,\ldots,y_m\}$. The result follows now by Lemma \ref{lem:sys2}.
\end{proof}

In the following, we show that the truth of $\Sigma^+_2$ formulae over $A$, and, consequently, of systems $(\mS,\mU)$ of sat- and unsat-equations with $\mU$ consisting only of negative word equations, is decidable. We need the following lemma.

\begin{lemma}\label{lem:trivialeqs}
Let $Y = \{ y_1,y_2,\ldots, y_n \} \subseteq X$ and let $U,V \in (Y \cup A)^*$. Let $k > |UV|$ and let $h : X^* \to A^*$ be the substitution such that $h(y_i) = \ta\tb^{k+i}\ta$. Then $h(U) = h(V)$ if and only if $U = V$ (the strings $U$ and $V$ coincide). 
\end{lemma}
\begin{proof}
The if direction is trivial. Suppose that $h(U) = h(V)$. Note that if $U$ is a proper prefix of $V$, then $h(U)$ is a proper prefix of $h(V)$ and vice-versa. Thus we can assume that $U$ is not a prefix of $V$ and $V$ is not a prefix of $U$. We shall proceed by induction on prefixes of $U$ and $V$. In particular, suppose that $U$ and $V$ have a common prefix $W$. Note that this holds for the base case $W = \varepsilon$. If $W = U = V$, we are done. Otherwise, since neither is a prefix of the other, there exist $x,x^\prime \in Y\cup A$ and $U^\prime, V^\prime$ such that $U = WxU^\prime$ and $V= Wx^\prime V^\prime$. This implies that $h(xU^\prime) = h(x^\prime V^\prime)$, and in particular that either $h(x)$ is a prefix of $h(x^\prime)$ or vice-versa. Wlog. consider the first case. If $x, x^\prime \in A$, then it is immediate that $x = x^\prime$. Similarly, since $\ta\tb^{k+i}\ta$ is only a prefix of $\ta\tb^{k+j}\tb$ if $i=j$, if $x,x^\prime \in Y$ it also follows that $x = x^\prime$. Since $h(x)$ is a prefix of $h(x^\prime)$, it cannot be that $x \in Y$ while $x^\prime \in A$, so the remaining case is when $x \in A$ and $x^\prime \in \{y_1,y_2,\ldots, y_n\}$, in which case we must have that $x = \ta$ since $\ta$ is the first letter of $h(x^\prime)$. Suppose this holds and let $U^\prime = zU^{\prime\prime}$ such that $z$ is the longest prefix of $U^\prime$ consisting only of terminal symbols. Then $h(U) = h(W)\ta z h(U^{\prime\prime})$ and $h(V) = h(W)\ta\tb^{k+i}\ta h(V^\prime)$ for some $\ta \in A$ and $i \in [1,n]$.  Since $h(y_j)$ starts with $\ta$ for all $j \in [1,n]$, either $h(U^{\prime\prime}) = \varepsilon$ or $h(U^{\prime\prime})$ has $\ta$ as a prefix. Consequently, $\tb^{k+i}$ must be a prefix of $z$. However, since $k>|UV| \geq |z|$, this is a contradiction, and we must have $x = x^\prime$. By induction, it follows that $U=V$.
\end{proof}

We can now show the next theorem.

\begin{theorem}\label{sigma2+}
The truth of $\Sigma^+_2$ formulae over $A^*$ is decidable.
\end{theorem}

\begin{proof}
Wlog.\,we may assume we have a sentence in disjunctive normal form as follows:
\begin{equation} \label{eqn:dnf}
  \begin{aligned}
	& \;\exists x_1,x_2,\ldots, x_n. \forall y_1,y_2,\ldots, y_m. (e_{1,1} \land \ldots \land e_{1,k_1}) \;\lor \\ 
    & (e_{2,1} \land \ldots \land e_{2,k_2}) \lor \ldots \lor (e_{t,1} \land \ldots \land e_{t,k_t}),
  \end{aligned}
\end{equation}
where $e_{i,j}$ are individual word equations over the variables $x_1,x_2,\ldots,x_n,y_1,y_2,\ldots,y_m$ and terminal symbols from $A$. By Lemma~\ref{lem:conjunctions}, we can, for each $i$, $1\leq i \leq t$, combine the equations $e_{i,1},e_{i,2},\ldots,e_{i,k_t}$ into a single equation $E_i$ without introducing any new variables. Thus, we get equations $E_1, E_2,\ldots, E_t$ such that (\ref{eqn:dnf}) is satisfiable if and only if 
\begin{equation} \label{eqn:disjunctions}
	\; \exists x_1,x_2,\ldots, x_n. \forall y_1,y_2,\ldots, y_m. E_1 \lor E_2 \lor \ldots \lor E_t 
\end{equation}
is satisfiable. 
Now, we claim that (\ref{eqn:disjunctions}) is satisfiable if and only if there exist values for $x_1,x_2,\ldots,x_n$ such that at least one of $E_1,E_2,\ldots,E_t$ becomes a trivial equation over the variables $y_1,y_2,\ldots,y_m$. It is clear that if such a substitution exists, the sentence is satsifiable. For the other direction, suppose that for any choice of $x_1,x_2,\ldots,x_n$, all the equations $E_i$, $1\leq i \leq t$ remain non-trivial (i.e. they are of the form $U_i = V_i$ with $U_i,V_i \in \{y_1,y_2,\ldots,y_m\}\cup A)^*$ such that $U_i \not = V_i$. Then by Lemma~\ref{lem:trivialeqs}, there exists a choice of $y_1,y_2,\ldots,y_m$ such that $U_i \not= V_i$ for all $i, 1\leq i \leq t$, and thus the sentence is false.
Therefore, to decide whether (2) is satisfiable, it is sufficient to decide, for each $i$, $1\leq i \leq t$, whether there exists a choice of $x_1,x_2,\ldots,x_n$ such that $E_i$ becomes a trivial equation. Suppose $E_i$ is the equation 
\[u_0 y_{i_1} u_1 y_{i_2} u_2 \ldots y_{i_p} u_p = v_0 y_{j_1} v_1 y_{j_2} \ldots y_{j_q} v_q\]
where $p,q\in \mathbb{N}_0$, $i_k,j_\ell \in [1,m]$ for $1 \leq k \leq p$ and $1\leq \ell \leq q$ respectively, and $u_k,v_\ell \in (\{x_1,x_2,\ldots,x_n\}\cup A)^*$ for $1\leq k \leq p$ and $1 \leq \ell \leq q$ respectively. Note that for a given choice of values for $x_1,x_2,\ldots,x_n$, the equation $E_i$ becomes trivial if and only if $p=q$, and $u_0 = v_0, u_1 = v_1, \ldots, u_p = v_p$. In other words, if $x_1,x_2,\ldots,x_n$ forms a solution to the system of equations $u_0 = v_0, u_1 = v_1, \ldots, u_p = v_p$ over the variables $x_1,x_2,\ldots,x_n$ and terminal symbols from $A$. It is well known that determining whether such a system has a solution is decidable (e.g., by Makanin's algorithm or by recompression) and hence the satisfiability of~(1) is decidable as required.\end{proof}

In conclusion, systems of sat- and unsat-equations exactly characterise the class of $\Sigma_2$ formulae over word equations. They strictly extend the fragment $\Sigma^+_2$, of positive $\Sigma_2$ formulae over word equations, which is decidable. It is interesting how the class of formulae that encode IPL can be compared to $\Sigma_2$ and $\Sigma^+_2$. For instance, are they strictly less powerful than $\Sigma_2$? 

Note that, as the fragment $\Sigma_2$ of $FO(A^*,\cdot)$ is undecidable, and, according to \cite{karhumaki2000}, every formula contained in this fragment can be expressed as a $\Sigma^+_3$ formula ($\exists\forall\exists$ quantified), it follows that the fragment $\Sigma^+_3$ of $FO(A^*,\cdot)$ is undecidable. Thus, Theorem \ref{sigma2+} is, in a sense, optimal.

\section{Undecidability results}
\label{sec:undec}

In the following section, we consider several undecidable extensions of word equations. 

Let $T_S$ denote the existential first-order two-sorted theory (with sorts $nat$ and $str$, respectively, for numbers and strings) consisting of string equations, a length function for strings, linear arithmetic over numbers, and a string-number conversion predicate (denoted as $strnum$). This predicate checks, for a given binary string $z$ and a number $x$, whether $z$ is the binary representation of $x$. $T_S$ is expressive enough that most string-related library functions from C, C++, Java, PHP, and JavaScript can be easily encoded in terms of its functions and predicates. 

Following B\"uchi and Senger, we define the first-order existential power arithmetic theory $T_P$ with the signature $\left< \mathbb{N}, 0, 1, +, P \right>$ where $P$ is a 3-ary relation defined by $(p, x, y)\in P$ if and only if $p = x \times 2^{y}$. We also define the predicate $P(\cdot,\cdot,\cdot)$ which returns true iff the argument-tuple belongs to $P$. In order to prove that $T_S$ is undecidable, we give a reduction from $T_P$, which was shown to be undecidable~\cite{buchi1990}.
\begin{theorem}[\citet{buchi1990}] \label{thm:TPUndecidable}
  $T_P$ is undecidable.
\end{theorem}

An easy corollary of Theorem~\ref{thm:TPUndecidable} is the undecidability of a
variant of $T_P$ (which we denote as $T_{P,bin}$) where all numbers
are represented in binary representation. All functions and predicates
of $T_P$ are easily reinterpreted appropriately in $T_{p,bin}$. In
particular, $P(p,x,y)$ can be interpreted as the
equation $p_{bin} = x_{bin}0^y$ where $p_{bin}$ is the binary representation of the number $p$
and $x_{bin}$ is the binary representation of the number $x$.

\begin{theorem}
  The satisfiability problem of the first-order existential theory
  $T_S$ is undecidable.
\end{theorem}
\begin{proof}
  We show that the decidability of $T_{P, bin}$ can be reduced to the
  decidability of $T_S$. Clearly, addition can be expressed
  in $T_S$. $P(p, x, y)$ is expressible in $T_S$ as follows:
  \begin{align*}
  \exists z:str. \exists x_s:str. &strnum(x_s, x) \land 0z = z0 \land len(z) = y \; \land \\
  &strnum(x_s \cdot z, p).
  \end{align*}
  
  In the above formula $0z$ (resp., $z0$) is the concatenation of $0$ and $z$ (resp. $z$ and $0$) as binary strings.
\end{proof}

Next, we show the undecidability of various extensions of the existential theory of word equations. In each case, undecidability is ultimately obtained by showing that, for a unary-style encoding of integers following~\cite{buchi1990} (where a number is represented using the length of a string in the form $a^* b$, so $\varepsilon$ is 0, $b$ is 1, etc.),
the additional predicate(s) can be used to derive a multiplication predicate $\Multiply(x,y,z)$ which decides for numbers $i,j,k$ encoded in this way (i.e., $x= \ta^{i-1}\tb, y=\ta^{j-1}\tb, z=\ta^{k-1}\tb$), whether $k = ij$. Since a corresponding addition predicate can easily be modelled for this encoding using only word equations, undecidabilty follows immediately.

The extensions are given as binary and 3-ary relations which may easily be interpreted as predicates.

\begin{definition}
Let $Eq_\ta$, $Eq_\tb$, $\Abelian$, $\Morphism$, $\Projection$, $\Subword \subset A^*\times A^*$ and $\Shuffle$, $\Insert$, $\Erase \subset A^*\times A^* \times A^*$ be the relations given by:
\begin{itemize}
\item $(x,y) \in Eq_\ta$ if and only if $|x|_\ta = |y|_\ta$, and $(x,y) \in Eq_\tb$ if and only if $|x|_\tb = |y|_\tb$,
\item $(x,y) \in \Abelian$ if and only if $x$ and $y$ are abelian-equivalent,
\item $(x,y) \in \Morphism$ if and only if there exists a morphism $h:A^* \to A^*$ such that $h(x) = y$,
\item $(x,y) \in \Projection$ if and only if there exists a projection $\pi: A^* \to A^*$ such that $\pi(x) = y$,
\item $(x,y) \in \Subword$ if and only if $x$ is a (scattered) subword of $y$.
\item $(x,y,z) \in \Shuffle$ if and only if $z \in x\shuf y$,
\item $(x,y,z) \in \Erase$ if and only if $z$ may be obtained from $x$ by removing some (or all) occurrences of $y$,
\item $(x,y,z) \in \Insert$ if and only if $z$ may be obtained from $x$ by inserting any number of occurrences of $y$.
\end{itemize}
For each of the above relations we can also define a predicate with the same name which returns true if the tuple of arguments belongs to the relation and false otherwise.
\end{definition}

Note that the membership problems for all the above relations are in NP, and therefore decidable. In some cases, our approach is simplified by reducing to predicates $Only{\ta}s(x,y)$ and $Only{\tb}s(x,y)$ which return true if and only if $y = \ta^{|x|_\ta}$ (respectively $y = \tb^{|x|_\tb}$). \citet{buchi1990} show how these predicates can easily be used to model multiplication, and thus undecidability follows.

\begin{theorem}[\citet{buchi1990}]
Given the predicates $Only{\ta}s(x,y)$ and $Only{\tb}s(x,y)$ it is possible to model multiplication. 
\end{theorem}

\begin{corollary}[\citet{buchi1990}]
The existential theory of word equations with additional predicates $Only{\ta}s(x,y)$ and $Only{\tb}s(x,y)$ is undecidable.
\end{corollary}

It is a straightforward observation that the predicates $Eq_\ta$ and $Eq_\tb$ which compare occurrences of a single letter are equivalent to $Only{\ta}s$ and $Only{\tb}s$ respectively in the sense that one can be used to model the other and vice versa.

\begin{proposition}\label{prop:countletters}
The predicate $Eq_\ta$ is equivalent to the predicate $Only{\ta}s$. Likewise, $Eq_\tb$ is equivalent to $Only{\tb}s$.
\end{proposition}
\begin{proof}
Given $Eq_\ta$, we can construct $Only{\ta}s$ as follows:
\begin{align*}
Only{\ta}s(x,y) := \;& y \ta = \ta y \land Eq_\ta(x,y).
\end{align*}
Given $Only{\ta}s$ we can construct $Eq_\ta$ as follows:
\begin{align*}
Eq_\ta(x,y) := \;& \exists z.\; Only{\ta}s(x,z) \land Only{\ta}s(y,z).
\end{align*}
The equivalence of $Only{\tb}s$ and $Eq_\tb$ can be shown in the same way.
\end{proof}

As a consequence, we have the following:

\begin{corollary}
The existential theory of word equations is undecidable when augmented with both the predicates $Eq_\ta, Eq_\tb$.
\end{corollary}

\citet{buchi1990} also showed that if only one of the predicates $Only{\ta}s$, $Only{\tb}s$ is allowed, but in addition also a  predicate $\Length(x,y)$ which evaluates to true if and only if $|x| = |y|$, then the theory also remains undecidable. Thus, the same holds when considering $Eq_\ta$ (or $Eq_\tb)$.

\begin{corollary}
The existential theory of word equations is undecidable when augmented with both the predicates $Eq_\ta$ and $\Length$.
\end{corollary}

It is worth noting that the case that only $Only{\ta}s$ (or equivalently any one of $Eq_\ta$, $Eq_\tb$ or $Only{\tb}s$) is given (i.e., without $\Length$), it remains unknown whether the theory is decidable. Next, we show that each of the other predicates can be used to obtain the predicates $Only{\ta}s$ and $Only{\tb}s$. For $\Subword$, undecidability was also shown by Haflon et al.~\cite{halfon2017}. 

\begin{proposition}\label{prop:abelianetc}
Given any of the predicates $\Abelian$, $\Shuffle$, $\Projection$, $\Subword$, $\Insert$, $\Erase$, it is possible to construct the predicates $Only{\ta}s$ and $Only{\tb}s$.
\end{proposition}
\begin{proof}
W.l.o.g.\,suppose $A = \{\ta_1,\ta_2,\ldots, \ta_n\}$ where $\ta_1 = \ta$ and $\ta_2 = \tb$. For each predicate, we shall give a construction for either $Only{\ta}s$ or $Eq_\ta$. In each case $Only{\tb}s$ or $Eq_\tb$ can be constructed in the same way mutatis mutandis.\\
\emph{Case 1. ($\Abelian$) }Suppose we have the predicate $\Abelian$. Then we can construct $Eq_\ta$ as follows:
\begin{align*}
 & Eq_{\ta}(x,y) :=\; \exists x^\prime,y^\prime, z_2,\ldots z_n, z^\prime_2, \ldots, z^\prime_n. \; z_2 \ta_2 = \ta_2 z_2 \;\land\\
 &z_3 \ta_3 = \ta_3 z_3 \land \ldots \land z_n \ta_n = \ta_n z_n \land \\ 
 & z_2^\prime \ta_2 = \ta_2 z^\prime_2  \land z^\prime_3 \ta_3 = \ta_3 z^\prime_3 \land \cdots \land z^\prime_n \ta_n = \ta_n z^\prime_n \land \; \\
 &x^\prime = xz_1z_2\cdots z_n \land y^\prime = y z_2^\prime z_3^\prime \cdots z_n^\prime \; \land \Abelian(x^\prime, y^\prime).
\end{align*}
By Lemma~\ref{lem:commute}, the first three lines are satisfied if and only if for $2 \leq i \leq n$, $z_i,z^\prime_i \in \{\ta_i\}^*$. It follows directly that there exist choices of $z_i,z_i^\prime$ such that $x^\prime$ and $y^\prime$ are abelian equivalent if and only if $|x|_\ta = |y|_\ta$.\\
\emph{Case 2. ($\Shuffle$) }Suppose we have the predicate $\Shuffle$. We construct the predicate $Only{\ta}s$ as follows:
\begin{align*}
&Only{\ta}s(x,y) := \;  \exists y_2, y_3,\ldots y_{n-1}, z_2,z_3, \ldots z_{n}.\; y \ta_1 = \ta_1 y\; \land \\
&z_2 \ta_2 = \ta_2 z_2 \land \ldots \land z_{n}\ta_{n} = \ta_{n} z_{n} \;\land 
\Shuffle(y,z_2,y_2) \land\\
&\Shuffle(y_2,z_3,y_3) \land \ldots  \land \Shuffle(y_{n-1}, z_n, x).
\end{align*} 
To verify the correctness, note that by Lemma~\ref{lem:commute}, the first and second lines are satisfied if and only if $z_i \in \{\ta_i\}^*$ for $2 \leq i \leq n$ and $y \in \{\ta_1\}^*$. The third line is also satisfied if, in addition, $y_i$ is obtained by shuffling $y_{i-1}$ with $z_i$. The net effect of this is that $y_i$ is obtained from $y_{i-1}$ by adding occurrences of $\ta_{i}$. It is straightforward that if $y = \ta_1^{|x|_{\ta_1}}$, then there exist choices of $z_2,z_3,\ldots,z_n$ and $y_2,\ldots y_{n-1}$ such that the sentence is true. Similarly, since the shuffles are only able to introduce the letters $\ta_i$ for $i \geq 2$, if $y \not=\ta_1^{|x|_{\ta_1}}$ then the sentence cannot be satisfied.\\
\emph{Case 3. ($\Projection$)} Suppose we have the predicate $\Projection$. We can construct the predicate $Only{\ta}s$ as follows:
\begin{align*}
Only{\ta}s(x,y) := \; & y \ta = \ta y \land  \Projection(x,y).
 \end{align*}
To verify the construction, suppose the sentence evaluates to true. Then by Lemma~\ref{lem:commute}, since $y$ satisfies $y\ta = \ta y$, it follows that $y \in \{\ta\}^*$. Moreover, since $y$ satisfies $\Projection(y,z)$, we must necessarily have $y = \ta^{|x|_\ta}$. It is straightforward to see that in the other direction, if $y=\ta^{|x|_\ta}$ (i.e. $Only{\ta}s(x,y)$ is true), then the sentence is satisfied.\\
\emph{Case 4. ($\Subword$)} Suppose we have the predicate $\Subword(x,y)$. Then we can construct the predicate $Eq_\ta$ as follows:
\begin{align*}
Only{\ta}s(x,y) :=\; & \exists z. y\ta = \ta y \land \Subword(y,x) \; \land\\ 
& z = y\ta \land \lnot \Subword(z,x). 
\end{align*}
To verify the correctness, it is sufficient to notice firstly that by Lemma~\ref{lem:commute}, $y\ta = \ta y$ if and only if $y \in \{\ta\}^*$, and secondly that this implies that $\Subword(z,x) \land \lnot \Subword(z\ta,x)$ also holds if and only if $y = \ta^{|x|_\ta}$.\\
\emph{Case 5 ($\Erase$)} Suppose we have the predicate $\Erase$. We construct the predicate $Only{\ta}s$ as follows:
\begin{align*}
Only{\ta}s(x,y) := \; & \exists z_2,\ldots, z_{n-1}.\; y \ta_1 = \ta_1 y \; \land\\
&\Erase(x, \ta_{n}, z_{n-1}) \land \Erase(z_{n-1}, \ta_{n-1}, z_{n-2}) \\ &\land \ldots \land \Erase(z_2, \ta_2, y).
\end{align*}
To verify the correctness, notice that for the sentence to be satisfied, $z_{n-1}$ must be obtained by erasing only $\ta_n$s from $x$, and in general $z_i$ must be obtained from $z_{i+1}$ by removing only $\ta_{i+1}$s. Moreover, by Lemma~\ref{lem:commute}, $y$ must consist only of $\ta_1$s, and must be obtained by removing only $\ta_2$s from $z_2$. The net effect of this is that $y$ must be the product of removing all occurrences of $\ta_2,\ta_3,\ldots,\ta_n$ from $x$ (i.e., $y = \ta_1^{|x|_{\ta_1}}$). Conversely, it is straightforward to see that if this holds, the sentence is satisfied. 
\\
\emph{Case 6. ($\Insert$)} Follows directly from the fact that $\Insert(x,y,z)$ is true if and only if $\Erase(z,y,x)$ is true, along with the result from Case~5.
\end{proof}

We discuss the predicate $\Morphism$ separately, and rather than reducing to $Eq_\ta$ and $Eq_\tb$, we construct the $\Multiply$ predicate directly.

\begin{proposition}\label{prop:morphisms}
Let $|A| \geq 3$. Then given the predicate $\Morphism$, it is possible to construct the predicate $\Multiply$.
\end{proposition}

\begin{proof}
Assume that $A$ contains at least three distinct letters: $\ta, \tb, \tc$. We shall construct a predicate $\Multiply_2(x,y,z)$ which returns true if $x = \ta^i \tb$, $y = \ta^j\tb$, $z = \ta^{ij}\tb$ and $ij \geq 2$. 
Note we can obtain $\Multiply$ from this, as $\Multiply(x, y, z) = \Multiply_2(ax, ay, az)$ for $x, y, z \neq \varepsilon$. For ease of exposition, we define first a predicate checking some `initial conditions':
\begin{align*}
& init(x,x',x'',y,y',z,z') := \exists w,w',w''.\; x'\! =\! w\ta \land  y'\!=\! w' \ta \land \\
& (x' = w'' \ta\ta \lor y' = w'' \ta\ta)\; \land
x' \ta = \ta x' \land
y' \ta = \ta y' \land \\
& z' \ta = \ta z' \land x = x' \tb \land y = y' \tb \land z = z \tb \; \land x'' x = x x''
\end{align*}
Recalling Lemma~\ref{lem:commute}, it is straightforward to see that $init$ evaluates to true if and only if there exist  $i,j,k,\ell, p \in \mathbb{N}_0$ with $ij \geq 2$ such that:
\begin{enumerate}
\item{} $x' = \ta^i$, $y' =\ta^j$, $z' = \ta^k$, and
\item{} $x = \ta^i\tb$, $y = \ta^j\tb$ $z = \ta^k\tb$, and
\item{} $x'' = (\ta^i\tb)^p$.
\end{enumerate}
Now we give the predicate for $\Multiply_2$ as follows:
\begin{align*} & \Multiply_2(x,y,z) := \; \exists x',x'',y',z',u,v.\;\\& init(x,x',x'',y,y',z,z') \; \land\\
& \Morphism(x'',y') \land \Morphism(y',x'') \; \land \Morphism(u,v) \land\\
& u = x'' \tc\tc x'' x' \tc\tc \tb \; \land v = z' \tc\tc z'x' \tc\tc.
\end{align*}
Suppose that Conditions~(1)-(3) are met (i.e., $init$ is satisfied). Consider the subclause 
$\Morphism(x'',y') \land \Morphism(y', x'').$ This is satisfied if and only if there exist morphisms $g,h :A^* \to A^*$ such that $g((\ta^i\tb)^p) = \ta^j$ and $h(\ta^j) = (\ta^i \tb)^p$. Clearly, the latter implies that $p$ is a multiple of $j$, while the former implies that $j$ is a multiple of $p$, and hence if both are satisfied then $j = p$. On the other hand, if $j=p$, then it is easy to construct such morphisms ($g$ maps $\tb$ to $\ta$ and $\ta$ to $\varepsilon$ while $h$ maps $\ta$ to $\ta^i\tb$). Thus this subclause is satisfied in addition to the $init$ predicate if and only if Conditions~(1)-(3) hold for $p = j$. By elementary substitutions, the last line is also satisfied if and only if
$ u = (\ta^i\tb)^j \tc\tc (\ta^i\tb)^j \ta^i \tc\tc \tb
\mbox{, and }v = (\ta^k \tc \tc \ta^{k+i} \tc\tc).$ 
It remains to show that there exists a morphism $f : A^* \to A^*$ such that $f(u) = v$ if and only if $k = ij$. In the case that $k = ij$, the morphism $f$ may be given e.g.\,by $f(\ta) = \ta$, $f(\tb) = \varepsilon$ and $f(\tc) = \tc$. For the other direction, assume that such a morphism $f$ exists. Firstly, consider the case that $f(\tc) \in \{\ta,\tb\}^*$. Then $\tc$ must occur in $f(\ta)$ or $f(\tb)$. However, under our assumption that $ij \geq 2$, this implies $|f(u)|_\tc > 4$ meaning $f(u) \not= v$ which is a contradiction. Consequently, we may infer that $f(\tc)$ contains the letter $\tc$. Then since $|u|_\tc = |v|_\tc$, it follows that $f(\tc) = v_1 \tc v_2$ where $v_1,v_2 \in \{\ta,\tb\}^*$. Thus $f(u) = f(\ta^i\tb)^j v_1 \tc v_2 v_1 \tc v_2 f(\ta^i\tb)^j \ta^i v_1 \tc v_2 v_1 \tc v_2 f(\tb)$. It follows that $v_1 = v_2 = \varepsilon$, and thus that $f(\tb) = \varepsilon$. Hence we must have that $f(\ta^{ij}) = \ta^k$ and $f(\ta^{ij+i}) = \ta^{k+i}$. Clearly, $f(\ta) = \ta^n$ for some $n \in \mathbb{N}$. Thus we have $nij = k$ and $nij + ni = k + i$. Hence, $n = 1$ and $k = ij$, as required.
\end{proof}

Summarising the consequences of Propositions~\ref{prop:countletters},~\ref{prop:abelianetc} and~\ref{prop:morphisms}, we have the following theorem.

\begin{theorem}
The existential theory of word equations becomes undecidable when augmented with any of the following predicates: $\Abelian$, $\Shuffle$, $\Projection$, $\Subword$, $\Morphism$ (if $|A|\geq 3$), $\Insert$, $\Erase$.
\end{theorem}

\section{Decidability with Restricted Form}
\label{sec:dec}

We shall now concentrate on decidable variants. In particular, we shall consider extensions to the theory of word equations over $A^*$ in conjunction with restrictions to the structure of allowed equations. Firstly, we note that if we allow at most one terminal symbol appearing in the equations (this is a weaker restriction than enforcing $|A|=1$), then the existential theory remains decidable when augmented with linear arithmetic over the lengths of variables. 

\begin{theorem}\label{the:oneletter}
Let $\ta \in A$. The satisfiability of quantifier-free positive formulae over word equations $U = V$, such that $U, V \in (X\cup\{\ta\})^*$, with linear length constraints is NP-complete.
\end{theorem}
\begin{proof}
First we consider a single equation $U=V$. 

Let us overload the notation $|U|_x$ to denote the number of
occurrences of the variable $x$ in $U \in (A \cup X)^*$.

Consider the equation $U(x_1, \dots, x_n) = V(x_1, \dots, x_n)$ that does not
contain any letters from the alphabet other than $\ta$. Then any solution
$h$ to this equation must satisfy $|h(U)| = |h(V)|$ which implies the
linear Diophantine equation
\begin{equation} \label{eqn:linear}
  \begin{aligned}
    & |U|_{x_1} |h(x_1)| + \dots + |U|_{x_n} |h(x_n)| + |U|_\ta \\
    = & |V|_{x_1} |h(x_1)| + \dots + |V|_{x_n} |h(x_n)| + |V|_\ta
  \end{aligned}
\end{equation}

If we consider only solutions $h$ where $h(x_i) \in \{\ta\}^*$ for all $1 \le i \le n$,
then the set of solutions are exactly the morphisms corresponding to the solutions
of equation~\ref{eqn:linear}.
Furthermore, any general solution must also satisfy equation~\ref{eqn:linear}, so
for every solution (which may involve elements of the alphabet other than $\ta$), there
is a solution using only $\ta$ with the same lengths for each variable.

Then to solve the satisfiability problem for a conjunction of equations with occurrences
of at most one letter $\ta$ with an additional set of linear length
constraints $\theta$ (i.e., a system of such word equations with length constraints), it is sufficient to check the satisfiability of the conjunction
of equation~\ref{eqn:linear} (for each equation) and $\theta$. This can be done
since each equation is linear.

If this system is satisfiable, a value for the length
of each variable can be obtained, and a solution using all $\ta$s can be constructed. Conversely, if a solution exists, the lengths of the variables
under this solution will be a solution to the system of linear equations since 
every solution to the equation must satisfy equation~\ref{eqn:linear}.

Now, each quantifier-free positive formula over word equations can be rewritten in disjunctive normal form, i.e., a disjunction of conjunctions of word equations. Deciding whether the entire formula is satisfiable is equivalent to deciding whether one of the conjunctions is satisfiable. This can be done as above.

It is clear that the system of equations~\ref{eqn:linear} augmented by the length constraints $\theta$ can be constructed in polynomial time. Solving systems of linear equations for non-negative integers is in NP~\cite{papadimitriou1981}, and thus the above algorithm runs in non-deterministic polynomial time. Conversely, it is easy to see that the linear length constraints can be turned into inequalities (for example, $|h(x)| \geq |h(y)|$ can be modelled with the equation $x = yz$ and the length constraint $|h(x)| = |h(y)|$). Thus the standard reduction from 3SAT to integer linear programming can be applied to get NP-hardness.
\end{proof}

Complementing the above result, we can show that the satisfiability of quantifier-free first order formulae over word equations $U = V$ (so including negation), such that $U, V \in (X\cup\{\ta\})^*$, with linear length constraints is equivalent to solving arbitrary word equations with length constraints. As such, we cannot say anything about the decidability of such formulae. One direction of our result is immediate, we only show the other one.

\begin{theorem}\label{the:permutation}
Let $|A|\geq 2$ and $\ta\in A$. Given an equation $U=V$, with $U,V\in (A\cup X)^*$, with linear length constraints~$\theta$, there exists a system $\mS$ of positive and negative equations $U_i =V_i$ or $U_i\neq V_i$ with $U_i,V_i\in (X'\cup\{\ta\})^*$ and $X\subset X'$, such that $\mS$ is satisfiable if and only if $U=V$ is satisfiable.
\end{theorem}
\begin{proof}
Let $A=\{\ta_1,\ldots,\ta_n\}$, with $\ta=\ta_1$. We define the set of variables $Y=\{y_1,\ldots,y_n\}$, such that $X\cap Y=\emptyset$. Now define the set of negative equations $S_1=\{y_i\neq y_j\mid 1\leq i<j\leq n\}$. Moreover, let $U'=V'$ be the equation obtained by replacing in $U=V$ each occurrence of $a_i$ by $y_i$, for $1\leq i\leq n$. Now, let $\mS$ be the system defined by $S_1\cup\{U'=V'\}\cup\{y_1=\ta_1\}$ with the length constraints defined by $\theta$ and $|y_i|=1$ for all $2\leq i \leq n$. Basically, the equations $S_1\cup\{y_1=a\}$ and the new length constraints ensure that $\{y_1,\ldots,y_n\}$ encode a permutation of $A$. As the actual label of the symbols of $A$ is not important to the satisfiability of $U=V$ (i.e., we can relabel the letters as we want, as long as we assign different labels to different letters), it follows that $U=V$ is satisfiable if and only if $\mS$ is satisfiable.
\end{proof}
If $|A|=1$, the satisfiability of quantifier-free first order formulae over word equations is decidable, as their theory can be seen as a fragment of the Presburger arithmetic.

Building on Theorem~\ref{the:oneletter}, the next result considers the $\Sigma_2$ fragment in the case that only one letter may appear in the equations (although recall that this does not imply that $|A| = 1$). In particular, if the positive theory only is considered, but in addition, the $\Length$ predicate defined in the previous section (i.e., $\Length(x,y)$ is true if and only if $|x| = |y|$) is allowed, then satisfiability remains decidable. 
Note in particular that the $\Length$ predicate can be used in conjunction with simple equations to model arbitrary linear length constraints.

\begin{theorem}\label{the:oneletter2}
Let $\ta \in A$. The positive $\Sigma_2$ fragment, restricted to word equations containing only the terminal symbol $\ta$, augmented with the $\Length$ predicate, is decidable.
\end{theorem} 
\begin{proof}
For the purposes of this proof we shall say that a term is trivial if, for all the word equations $U = V$, $U$ and $V$ are identical, and moreover, all $\Length$ predicates take identical arguments (i.e. they are of the form $\Length(z,z)$).
If $|A| = 1$, decidability follows from the decidability of Presburger arithmetic by the same arguments as in the proof of Theorem~\ref{the:oneletter}. Thus we may assume $\ta,\tb \in A$ with $\ta \not= \tb$.
W.l.o.g.\,we may assume that we have a sentence in disjunctive normal form as follows:
\begin{equation} \label{eqn:dnf2}
  \begin{aligned}
	& \;\exists x_1,x_2,\ldots, x_n. \forall y_1,y_2,\ldots, y_m. (e_{1,1} \land \ldots \land e_{1,k_1}) \;\lor \\ 
    & (e_{2,1} \land \ldots \land e_{2,k_2}) \lor \ldots \lor (e_{t,1} \land \ldots \land e_{t,k_t}),
  \end{aligned}
\end{equation}
where the $e_{i,j}$ are either:
\begin{enumerate}
\item{}of the form $\Length(z_1,z_2)$ where\\ $z1,z2 \in \{x_1,x_2,\ldots,x_n,y_1,y_2,\ldots,y_m\} \cup A^*$, or
\item{}individual word equations over the variables \\$x_1,x_2,\ldots,x_n,y_1,y_2,\ldots,y_m$ and the terminal symbol~$\ta$.
\end{enumerate}

As with the proof of Theorem~\ref{sigma2+}, we shall show that an assignment for $x_1,x_2,\ldots,x_n$ satisfies~(\ref{eqn:dnf2}) if and only if there exists $s,1\leq s \leq t$ such that all the resulting atoms $e_{s,i}$ become trivial. The `if' direction is straightforward, thus we consider the `only if' direction. Suppose the $x_1,x_2,\ldots,x_n$ are fixed, and consider the result of each $e_i$ under the substitution. Suppose that for each  $s, 1\leq s \leq t$ there exists $r_s, 1\leq r \leq k_s$ such that $e_{s,r_s}$ is non-trivial. Let $p$ be the maximum over the lengths of all constant terms in the sentence, lengths of the $x_i$, and lengths of equations given by the type-(2) atoms $e_{i,j}$ for $i, 1\leq i \leq t$, $1 \leq j \leq k_i$. Consider the choice of $y_1,y_2,\ldots,y_m$ given by $y_i = \ta\tb^{p +i}\ta$. By Lemma~\ref{lem:trivialeqs}, if $e_{s,r_s}$ is of type (2), then it will evaluate to false. If $e_{s,r_s}$ is of type (1), then we have three cases. Firstly, if both arguments to the $\Length$ predicate are constant terms in $A^*$, then clearly $e_{s,r_s}$ will evaluate to false since it is non-trivial. Similarly, since the $y_i$ are longer than all constant terms, if exactly one of the arguments is a constant in $A^*$ while the other is a variable in $\{y_1,y_2,\ldots,y_m\}$, then $e_{s,r_s}$ will also evaluate to false. Finally, since $|y_i| \not= |y_j|$ for all $i \not=j$, if both arguments are variables $e_{s,r_s}$ will again evaluate to false. Summarising the above, for any given choice of $x_1,x_2,\ldots,x_n$ there exists single a choice of $y_1,y_2,\ldots,y_m$ such that any of the conjunctions containing a non-trivial equation or $\Length$ predicate will be false. It follows that the sentence is satisfiable if and only if there exists a choice for $x_1,x_2,\ldots,x_n$ and $s,1\leq s \leq t$ such that all the $e_{s,i}$ terms, $1\leq i \leq k_s$ become trivial.

We have shown already in the proof of Theorem~\ref{sigma2+} that for terms $e_{i,j}$ of type (2), this is reduced to solving a series of existentially quantified word equations over $x_1,x_2,\ldots,x_n$. Moreover, a term $e_{i,j}$ of type (1) may only become trivial under some substitution for the $x_i$s either if it is already trivial, in which case it can just be removed, or if both arguments are in $\{x_1,x_2,\ldots,x_n\}$. Thus, any of the clauses  
$(e_{i,1} \land \ldots \land e_{i,k_i})$ 
containing a term $e_{i,j}$ not conforming to these two cases can be removed entirely. After these two phases of removal, it remains to solve, for each $s$, $1\leq s\leq t$, a series of systems of equations (derived from the $e_{s,i}$ terms of type (2), as described in the proof of Theorem~\ref{sigma2+}) subject to a system of linear length constraints (derived from the terms of type (1)). It is clear that the resulting equations will also only contain the terminal symbol $\ta$, since they are taken directly from the original equations, so the decidability follows from Theorem~\ref{the:oneletter}.
\end{proof}

Note that Theorems~\ref{the:oneletter2} and~\ref{the:permutation} together do not imply decidability of the existential theory of word equations with length constraints, due to the fact that the former excludes the use of logical negation while the latter requires it. On the other hand, it follows from Theorem~\ref{the:permutation}, along with the fact that the full $\Sigma_2$ fragment -- in which negation is allowed -- is undecidable, that the full $\Sigma_2$ fragment with $\Length$ but restricted to equations with only one terminal symbol is undecidable. Therefore, the decidability shown in Theorem~\ref{the:oneletter2} is, in a sense, optimal. 

If instead of restricting the terminal symbols appearing in the equation(s) we restrict the variables, we are also able to obtain decidability when augmenting the theory with both linear arithmetic over variable lengths, and regular constraints given in the form of DFAs.

\begin{theorem}\label{the:constraintsNP}
The satisfiability of strictly regular-ordered word equations with linear length constraints and regular constraints given by DFAs is NP-complete.
\end{theorem}
First we need the following lemma:
\begin{lemma}\label{lem:intersection}
Let $L$ be a regular language given by a DFA, $M$, with $n$ states. Let $\alpha,\beta \in A^*$. Then there exist $q\in \mathbb{N}$, $P, S \subseteq \mathbb{N}_{\leq n}$  such that the intersection of $(\alpha\beta)^+\alpha$ and $L$ is given by
\[ \{(\alpha\beta)^s\alpha \mid s \in S\} \cup \{(\alpha\beta)^{q\mu + p}\alpha \mid \mu \in \mathbb{N} \land p \in P\}. \]
\end{lemma}
\begin{proof}
Suppose firstly that there does not exist $t>n$ such that $(\alpha\beta)^t\alpha \in L$. Then the claim follows directly with $S = \{s\in\mathbb{N} \mid (\alpha\beta)^s \alpha \in L\}$ and $P = \emptyset$. Now suppose instead that there exists a word $w = (\alpha\beta)^t \alpha \in L$ such that $t > n$. For $1 \leq i \leq t$, let $a_i$ be the state $M$ is in after reading the input $(\alpha\beta)^i\alpha$. Since $M$ has only $n$ states, there must exist $p_0, q$ with $p_0 < p_0 + q \leq n$ such that $a_{p_0} = a_{p_0 + q}$. Let $P \subseteq \mathbb{N}_{\leq n}$ such that
 \[\{a_{p_i} \mid p\in P\} = \{a_j \mid p_0 < j \leq p_0 + q \land a_j \text{ is accepting}\}.\]
Hence, for $p_0 < r < p_0 + q$, $(\alpha\beta)^r \alpha \in L$ if and only if $r \in P$. Moreover, since $M$ is deterministic, we have that for all $r \geq p_0 + q$, $M$ is in the same state after reading $(\alpha\beta)^r \alpha$ and $(\alpha\beta)^{r-q}\alpha$. It follows by induction that $(\alpha\beta)^r \alpha \in L$ if and only if there exists $\mu \in \mathbb{N}$, and $p\in P$ such that $r = \mu q + p$. Since $p_0 < n$, the statement follows directly.
\end{proof}

We are now ready to prove the main theorem. 
\begin{proof}
NP-hardness follows from the fact that satisfiability of strictly regular-ordered word equations without length constraints is NP-hard~\cite{MFCS2017}. Thus it remains to show inclusion in NP. Let $E$ be a strictly regular-ordered word equation $U=V$ with a set of linear length constraints $\theta$ and regular constraints $L_x$ for each variable $x \in X$. For convenience, we shall call the solutions to the equation $U = V$ ignoring length or regular constraints \emph{basic solutions}. Similarly, we shall refer to solutions to the equation satisfying the regular constraints, but ignoring the length constraints \emph{intermediate} solutions. The majority of the proof shall consider the structure of basic and intermediate solutions. We begin with basic solutions.

Let $U_x$ be the prefix of $U$ up to and including the first (and only) occurrence of $x$ for each $x\in X$, and let $\overline{U_x}$ be the prefix of $U$ up to and not including the first occurrence of $x$. Define $V_x$ and $\overline{V_x}$ similarly. For any variable $x$, note that since $U_x$ and $V_x$ contain exactly the same variables, and thus for any solution $h$, the difference in the lengths of $h(U_x)$ and $h(V_x)$ is exactly the difference in the sum of the lengths of the terminal words. In particular, this implies that $||h(U_x)|-|h(V_x)|| < |UV|$, and similarly, that $||h(\overline{U_x})|-|h(\overline{V_x})|| < |UV|$. Consequently, $h(x)$ can only be longer than $|UV|$ if the two occurrences `overlap', meaning that either
$$|h(\overline{U_x})| \leq |h(\overline{V_x})| \leq |h(U_x)|,\mbox{ or }|h(\overline{V_x})| \leq |h(\overline{U_x})| \leq |h(V_x)|.$$
From now on, we shall distinguish between \emph{overlapping} and \emph{non-overlapping} variables. Let $h$ be a basic solution to $U=V$. Given an overlapping variable $x$, let $p(x)$ be the minimal period of $h(x)$. It follows directly from the definition of a period that $p(x) \leq ||h(\overline{U_x})| - |h(\overline{V_x})|| \leq |UV|$, and moreover, that there exist $\alpha,\beta\in A^*$ such that $p(x) = |\alpha\beta|$ and $h(x) \in (\alpha\beta)^+\alpha$. It is also not difficult to see that for all $n\in \mathbb{N}$, the morphism $h^\prime$ given by $h^\prime(y) = h(y)$ for all $y \not= x$ and $h^\prime(x) = (\alpha\beta)^n\alpha $ is also a solution.

Consequently, all basic solutions to the equation are described by short (linear in $|UV|$) words and numerical parameters. More precisely, the set of basic solutions is given by finitely many `parametric' solutions $h$ of the form $h(x) = (\alpha_x\beta_x)^{n_x} \alpha_x$ where $|\alpha_x\beta_x|\leq |UV|$ and $n_x =0 $ if $x$ is not overlapping, and is a parameter taking any value in $\mathbb{N}$ otherwise.

Thus, in order to describe the intermediate solutions, we consider the possible values of $n_x$ for which $(\alpha_x\beta_x)^{n_x} \alpha_x \in L_x$. In the case that $x$ is non-overlapping, this is straightforward: the set is either $\{0\}$ or $\emptyset$. If $x$ is overlapping, we simply have to consider the intersection $(\alpha_x\beta_x)^+ \alpha_x \cap L_x$. In particular, we can easily compute $S_x$, $P_x$ and $q_x$ from Lemma~2 (i.e., $S$, $P$ and $q$ in the lemma) in polynomial time. Let the set of possible values for $n_x$ such that $h(x) \in L$ is be denoted by 
\[ \Delta_x = S_x \cup \{ \mu q_x + p \mid \mu \in \mathbb{N} \land p \in P_x\}.\]

We can now give a nondeterministic algorithm for solving the equation with linear length constraints and regular constraints as follows. Firstly, we guess which variables are overlapping. For each variable $x$, we then guess $\alpha_x$ and $\beta_x$, followed by whether $x$ is overlapping, and if so, compute $q_x$, $S_x$ and $P_x$. If the morphism $h$ given by $h(x) = (\alpha_x\beta_x)^n_x \alpha_x$, where $n_x = 0$ if $x$ is non-overlapping and $n_x = 1$ otherwise, is not a basic solution to $U=V$, then output no and we are done. Similarly, if there exists an overlapping variable $x$ such that $S_x \cup P_x = \emptyset$, or if there exists a non-overlapping variable $x$ such that $\alpha_x \notin L_x$, then output no and we are done. Otherwise it remains to determine whether there exist values of $n_x$ for each overlapping variable $x$ such that the length and regular constraints are both satisfied. To do this, we guess either an $s_x \in S_x$ or $p_x \in P_x$ for each overlapping variable $x$. Then, we construct a system of linear Diophantine equations from the set $\theta$ of linear length constraints by swapping each occurrence of $|h(x)|$ with $ |\alpha_x\beta_x|(p_x + q_x\mu_x) + |\alpha_x|$ if $x$ is overlapping, and $|\alpha_x|$ otherwise. Note that the result is a linear Diophantine system over variables $\mu_x$ for each overlapping variable $x$. Moreover, by definition, for every possible positive integer value of $\mu_x$, we can construct an intermediate solution to our equation (i.e. one which satisfies both $U = V$ and the regular constraints). Thus, there exists a solution satisfying the equation and all constraints (regular and length) if and only if there exists a non-negative solution (i.e., all unknowns are given non-negative values) to the linear Diophantine system. It follows from~\cite{papadimitriou1981} that if such a solution exists, then there is guaranteed to be solution for which the values are at most exponentially large and thus have polynomially sized binary encodings. Accordingly we can just guess the solution to the system and verify it in polynomial time. If such a solution exists, then we can output yes (and also return a compressed description of the solution) and if no solution exists, then we can output no and we are done.
\end{proof}

On the other hand, for regular-ordered equations without the strictness (i.e. variables may occur in only one side), the equivalent of Theorem~\ref{the:constraintsNP} does not hold. It is a straightforward exercise that regular-ordered equations where each side has only one singly-occurring variable, along with regular constraints given by DFAs is PSPACE-complete. This follows from the fact that determining whether the intersection of $n$ DFAs is empty is PSPACE-hard. 

The decidability of non-strict regular ordered equations with linear length constraints also appears to be harder, as the form of the augmented system of Diophantine equations does not necessarily need to be linear any more. In particular, the presence of variables occurring only on one side allows for variables, or parts of variables to be `ungrounded', in the sense that they can be substituted with any factor and the result remains a valid solution. For example, consider as a simple example the equation $x\ta\tb z = z y$. In the case of solutions $h$ such that $|h(z)| = |h(x)| + 2$, the possibilities for $h(z)$ are given by repetitions of $h(x)\ta\tb$. Hence $|h(z)| = n_z(2 + |h(x)|)$ for some $n_z \in \mathbb{N}$. However, we may choose $h(x)$ freely (although this will of course fix $h(y)$). Thus we can consider the length of $h(x)$ also as an unknown and the previous equation is not linear. 
\section{Conclusion}
\label{sec:hierarchy}

{\small 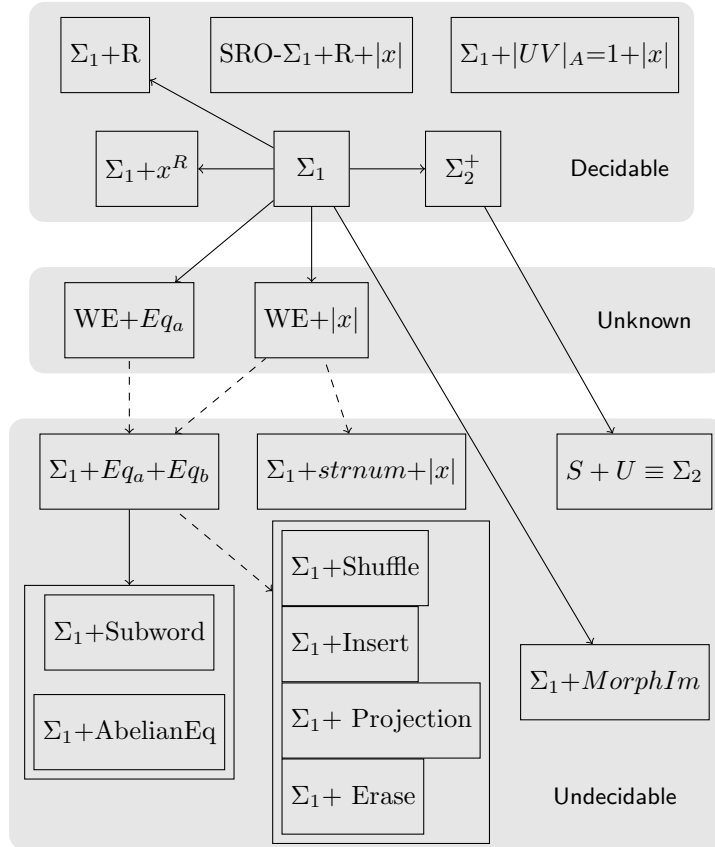
\begin{figure}
  \centering
  \begin{tikzpicture}
  	\tikzstyle{every node}=[draw, shape=rectangle, minimum size=1cm]
  	\tikzstyle{label}=[draw=none, minimum size=1cm]

    \node (SRO) {SRO-$\Sigma_1$+R+$|x|$};
    \node (WER) [left=0.8cm of SRO] {$\Sigma_1$+R};
    \node (WE) [below=0.5cm of SRO] {$\Sigma_1$};
    \node (ONE) [right=0.5cm of SRO] {$\Sigma_1$+$|UV|_A$=$1$+$|x|$};
    \node (rev) [left=of WE] {$\Sigma_1$+$x^R$};
    \node (S2+) [right=of WE] {$\Sigma^+_2$};
    \draw[<-] (WER) -- (WE);
    \draw[<-] (rev) -- (WE);
    \draw[<-] (S2+) -- (WE);

    \node (len) [below=of WE] {WE+$|x|$};
    \node (len0) [left=0.8cm of len] {WE+$Eq_a$};
    \draw[<-] (len) -- (WE);
    \draw[<-] (len0) -- (WE);

    \node (lenall) [below=of len0] {$\Sigma_1$+$Eq_a$+$Eq_b$};
    \node (pow2) [right=0.5cm of lenall] {$\Sigma_1$+$strnum$+$|x|$};
    \node[matrix,below=of lenall] (matrix) {
    & \node (replace) {$\Sigma_1$+Subword}; \\
    & \node (FST) [below=-0.7cm of replace]{$\Sigma_1$+AbelianEq}; \\
    };
    \node[matrix,right=0.5cm of matrix] (matrix2) {
    & \node (shuff) {$\Sigma_1$+Shuffle}; \\ 
    & \node (insert) {$\Sigma_1$+Insert}; \\
    & \node (proj){$\Sigma_1$+ Projection}; \\
    & \node (erase) {$\Sigma_1$+ Erase};  \\
    };
     \node (S+U) [right=1.2cm of pow2] {$S+U \equiv \Sigma_2$};
    \node (morph) [right=0.4cm of matrix2] {$\Sigma_1$+$MorphIm$};
      \draw[<-, dashed] (lenall) -- (len);
      \draw[<-, dashed] (lenall) -- (len0);
    \draw[<-, dashed] (pow2) -- (len);
    \draw[<-] (matrix) -- (lenall);
    \draw[<-, dashed] (matrix2) -- (lenall);
    \draw[<-] (morph) -- (WE);
    \draw[<-] (S+U) -- (S2+);

    \node[label] (undec) [below=0.5cm of morph] {\small \textsf{Undecidable}};
	\path let \p1 = (undec.east) in let \p2 = (WER) in node[label, anchor=east] at (\x1-0.1cm,\y2-1.5cm) {\small \textsf{Decidable}};
	\path let \p1 = (undec.east) in let \p2 = (len) in node[label, anchor=east] at (\x1+0.2cm,\y2) {\small \textsf{Unknown}};

    \begin{pgfonlayer}{background}
      \filldraw [line width=4mm,join=round,black!10]
      (WER.north -| ONE.east)  rectangle (S2+.south -| replace.west)
      (len.north -| morph.east)  rectangle (len.south -| replace.west)
      (lenall.north -| morph.east)  rectangle (erase.south -| matrix.west);
    \end{pgfonlayer}

  \end{tikzpicture}
  \caption{Reductions between different extensions of word equations. An arrow to a theory indicates that a reduction to it exists. A solid arrow indicates that no reduction in the other direction is possible, while a dashed arrow indicates that whether a reduction exists in other direction (i.e.\ whether an isomorphism exists) is unknown.}
  \label{fig:redgraph}
\end{figure}
}

In this paper we showed a series of decidability and undecidability results for various fragments 
of $FO(A^*,\cdot)$ and its extensions, starting from the theory of word equations. 
Our results are summarized and compared to some known results in Figure~\ref{fig:redgraph}. In that figure, on top of the usual notations of this paper, R stands for regular constraints; $|x|$ for length constraints, and $x^R$ for equations with reversal function; $SRO$ stands for strictly regular-ordered equations; $|UV|_A=1$ for equations with only one terminal symbol; $S+U$ stands for systems of sat- and unsat-equations.

From our results one can also immediately derive a series of already known results. The theory of word equations with an operator that replaces all occurrences of one string with another is undecidable according to \cite{lin2016}, where a reduction from PCP was shown; an alternate proof of this can be obtained using the $Erase$ operator we defined. Extending the theory of word equations by adding finite-state transducers also leads to undecidability, according to \cite{morvan2000}; again, we can use, e.g., the $Erase$ operator to obtain an alternate proof of the undecidabiliy of this theory. 

As future work, besides the main outstanding open problem of deciding whether word equations with length constraints are decidable, we also think that it is worth settling whether the satisfiability of an arbitrary $\Sigma_2$ formula can be reduced to the satisfiability of a formula corresponding to an instance of the Inclusion of Pattern Languages problem, or not. Also, settling the decidablity of the satisfiability problem for other classes of restricted word equations with length constraints (e.g., quadratic equations) seems appealing to us.
Whether Proposition~\ref{prop:morphisms} holds also for binary alphabets seems also interesting to us.



\bibliography{ms}


\begin{thebibliography}{32}


\ifx \showCODEN    \undefined \def \showCODEN     #1{\unskip}     \fi
\ifx \showDOI      \undefined \def \showDOI       #1{#1}\fi
\ifx \showISBNx    \undefined \def \showISBNx     #1{\unskip}     \fi
\ifx \showISBNxiii \undefined \def \showISBNxiii  #1{\unskip}     \fi
\ifx \showISSN     \undefined \def \showISSN      #1{\unskip}     \fi
\ifx \showLCCN     \undefined \def \showLCCN      #1{\unskip}     \fi
\ifx \shownote     \undefined \def \shownote      #1{#1}          \fi
\ifx \showarticletitle \undefined \def \showarticletitle #1{#1}   \fi
\ifx \showURL      \undefined \def \showURL       {\relax}        \fi
\providecommand\bibfield[2]{#2}
\providecommand\bibinfo[2]{#2}
\providecommand\natexlab[1]{#1}
\providecommand\showeprint[2][]{arXiv:#2}

\bibitem[\protect\citeauthoryear{Abdulla, Atig, Chen, Hol{\'{\i}}k, Rezine,
  R{\"{u}}mmer, and Stenman}{Abdulla et~al\mbox{.}}{2015}]%
        {abdulla2015}
\bibfield{author}{\bibinfo{person}{P.~A. Abdulla}, \bibinfo{person}{M.~F.
  Atig}, \bibinfo{person}{Y. Chen}, \bibinfo{person}{L. Hol{\'{\i}}k},
  \bibinfo{person}{A. Rezine}, \bibinfo{person}{P. R{\"{u}}mmer}, {and}
  \bibinfo{person}{J. Stenman}.} \bibinfo{year}{2015}\natexlab{}.
\newblock \showarticletitle{Norn: An {SMT} Solver for String Constraints}. In
  \bibinfo{booktitle}{{\em Proc. {CAV} 2015}} {\em (\bibinfo{series}{LNCS})},
  Vol.~\bibinfo{volume}{9206}. \bibinfo{pages}{462--469}.
\newblock


\bibitem[\protect\citeauthoryear{Aydin, Bang, and Bultan}{Aydin
  et~al\mbox{.}}{2015}]%
        {aydin2015}
\bibfield{author}{\bibinfo{person}{A. Aydin}, \bibinfo{person}{L. Bang}, {and}
  \bibinfo{person}{T. Bultan}.} \bibinfo{year}{2015}\natexlab{}.
\newblock \showarticletitle{Automata-Based Model Counting for String
  Constraints}. In \bibinfo{booktitle}{{\em Proc. {CAV} 2015}} {\em
  (\bibinfo{series}{LNCS})}, Vol.~\bibinfo{volume}{9206}.
  \bibinfo{pages}{255--272}.
\newblock


\bibitem[\protect\citeauthoryear{Barrett, Conway, Deters, Hadarean,
  Jovanovi{\'c}, King, Reynolds, and Tinelli}{Barrett et~al\mbox{.}}{2011}]%
        {barrett2011}
\bibfield{author}{\bibinfo{person}{C. Barrett}, \bibinfo{person}{C.~L. Conway},
  \bibinfo{person}{M. Deters}, \bibinfo{person}{L. Hadarean},
  \bibinfo{person}{D. Jovanovi{\'c}}, \bibinfo{person}{T. King},
  \bibinfo{person}{A. Reynolds}, {and} \bibinfo{person}{C. Tinelli}.}
  \bibinfo{year}{2011}\natexlab{}.
\newblock \showarticletitle{CVC4}. In \bibinfo{booktitle}{{\em Proc. CAV 2011}}
  {\em (\bibinfo{series}{LNCS})}, Vol.~\bibinfo{volume}{6806}.
  \bibinfo{pages}{171--177}.
\newblock


\bibitem[\protect\citeauthoryear{Berzish, Ganesh, and Zheng}{Berzish
  et~al\mbox{.}}{2017}]%
        {berzish2017}
\bibfield{author}{\bibinfo{person}{M. Berzish}, \bibinfo{person}{V. Ganesh},
  {and} \bibinfo{person}{Y. Zheng}.} \bibinfo{year}{2017}\natexlab{}.
\newblock \showarticletitle{ZSstrS: {A} string solver with theory-aware
  heuristics}. In \bibinfo{booktitle}{{\em Proc. {FMCAD} 2017}}.
  \bibinfo{publisher}{{IEEE}}, \bibinfo{pages}{55--59}.
\newblock


\bibitem[\protect\citeauthoryear{Bremer and Freydenberger}{Bremer and
  Freydenberger}{2012}]%
        {dominik}
\bibfield{author}{\bibinfo{person}{J. Bremer} {and} \bibinfo{person}{D.~D.
  Freydenberger}.} \bibinfo{year}{2012}\natexlab{}.
\newblock \showarticletitle{Inclusion problems for patterns with a bounded
  number of variables}.
\newblock \bibinfo{journal}{{\em Inf. Comput.\/}}  \bibinfo{volume}{220}
  (\bibinfo{year}{2012}), \bibinfo{pages}{15--43}.
\newblock


\bibitem[\protect\citeauthoryear{B{\"u}chi and Senger}{B{\"u}chi and
  Senger}{1990}]%
        {buchi1990}
\bibfield{author}{\bibinfo{person}{J.~R. B{\"u}chi} {and} \bibinfo{person}{S.
  Senger}.} \bibinfo{year}{1990}\natexlab{}.
\newblock \showarticletitle{Definability in the existential theory of
  concatenation and undecidable extensions of this theory}.
\newblock In \bibinfo{booktitle}{{\em The Collected Works of J. Richard
  B{\"u}chi}}. \bibinfo{publisher}{Springer}, \bibinfo{pages}{671--683}.
\newblock


\bibitem[\protect\citeauthoryear{Budkina and Markov}{Budkina and
  Markov}{1973}]%
        {Budkina}
\bibfield{author}{\bibinfo{person}{L.~G. Budkina} {and} \bibinfo{person}{A.~A.
  Markov}.} \bibinfo{year}{1973}\natexlab{}.
\newblock \showarticletitle{{F}-semigroups with three generators}.
\newblock \bibinfo{journal}{{\em Mat. Zametki\/}}  \bibinfo{volume}{14}
  (\bibinfo{year}{1973}), \bibinfo{pages}{267--277}.
\newblock


\bibitem[\protect\citeauthoryear{Day, Manea, and Nowotka}{Day
  et~al\mbox{.}}{2017}]%
        {MFCS2017}
\bibfield{author}{\bibinfo{person}{J.~D. Day}, \bibinfo{person}{F. Manea},
  {and} \bibinfo{person}{D. Nowotka}.} \bibinfo{year}{2017}\natexlab{}.
\newblock \showarticletitle{The Hardness of Solving Simple Word Equations}. In
  \bibinfo{booktitle}{{\em Proc. {MFCS} 2017}} {\em
  (\bibinfo{series}{LIPIcs})}, Vol.~\bibinfo{volume}{83}.
  \bibinfo{pages}{18:1--18:14}.
\newblock


\bibitem[\protect\citeauthoryear{Diekert, Je{\.z}, and Plandowski}{Diekert
  et~al\mbox{.}}{2016}]%
        {diekert2016}
\bibfield{author}{\bibinfo{person}{V. Diekert}, \bibinfo{person}{A. Je{\.z}},
  {and} \bibinfo{person}{W. Plandowski}.} \bibinfo{year}{2016}\natexlab{}.
\newblock \showarticletitle{Finding all solutions of equations in free groups
  and monoids with involution}.
\newblock \bibinfo{journal}{{\em Inf. Comput.\/}}  \bibinfo{volume}{251}
  (\bibinfo{year}{2016}), \bibinfo{pages}{263--286}.
\newblock


\bibitem[\protect\citeauthoryear{Durnev}{Durnev}{1995}]%
        {Durnev}
\bibfield{author}{\bibinfo{person}{V.~G. Durnev}.}
  \bibinfo{year}{1995}\natexlab{}.
\newblock \showarticletitle{Undecidability of the positive
  $\forall\exists$-theory of a free semigroup}.
\newblock \bibinfo{journal}{{\em Sib. Math. J.\/}}  \bibinfo{volume}{36.5}
  (\bibinfo{year}{1995}), \bibinfo{pages}{917–929}.
\newblock


\bibitem[\protect\citeauthoryear{Ganesh, Minnes, Solar{-}Lezama, and
  Rinard}{Ganesh et~al\mbox{.}}{2013}]%
        {Vijay_HVC}
\bibfield{author}{\bibinfo{person}{V. Ganesh}, \bibinfo{person}{M. Minnes},
  \bibinfo{person}{A. Solar{-}Lezama}, {and} \bibinfo{person}{M.~C. Rinard}.}
  \bibinfo{year}{2013}\natexlab{}.
\newblock \showarticletitle{Word Equations with Length Constraints: What's
  Decidable?}. In \bibinfo{booktitle}{{\em {HVC} 2012, Revised Selected
  Papers}} {\em (\bibinfo{series}{LNCS})}, Vol.~\bibinfo{volume}{7857}.
  \bibinfo{pages}{209--226}.
\newblock


\bibitem[\protect\citeauthoryear{Halfon, Schnoebelen, and Zetzsche}{Halfon
  et~al\mbox{.}}{2017}]%
        {halfon2017}
\bibfield{author}{\bibinfo{person}{S. Halfon}, \bibinfo{person}{P.
  Schnoebelen}, {and} \bibinfo{person}{G. Zetzsche}.}
  \bibinfo{year}{2017}\natexlab{}.
\newblock \showarticletitle{Decidability, complexity, and expressiveness of
  first-order logic over the subword ordering}. In \bibinfo{booktitle}{{\em
  Proc. {LICS} 2017}}. \bibinfo{publisher}{{IEEE} Computer Society},
  \bibinfo{pages}{1--12}.
\newblock


\bibitem[\protect\citeauthoryear{Harju and Nowotka}{Harju and Nowotka}{2003}]%
        {DirkTCS2003}
\bibfield{author}{\bibinfo{person}{T. Harju} {and} \bibinfo{person}{D.
  Nowotka}.} \bibinfo{year}{2003}\natexlab{}.
\newblock \showarticletitle{On the independence of equations in three
  variables}.
\newblock \bibinfo{journal}{{\em Theor. Comput. Sci.\/}} \bibinfo{volume}{307},
  \bibinfo{number}{1} (\bibinfo{year}{2003}), \bibinfo{pages}{139--172}.
\newblock


\bibitem[\protect\citeauthoryear{Hilbert}{Hilbert}{1900}]%
        {hilbert1900}
\bibfield{author}{\bibinfo{person}{D. Hilbert}.}
  \bibinfo{year}{1900}\natexlab{}.
\newblock \showarticletitle{Mathematische probleme}.
\newblock \bibinfo{journal}{{\em Nachrichten von der Gesellschaft der
  Wissenschaften zu G{\"o}ttingen, Mathematisch-Physikalische Klasse\/}}
  \bibinfo{volume}{1900} (\bibinfo{year}{1900}), \bibinfo{pages}{253--297}.
\newblock


\bibitem[\protect\citeauthoryear{Je{\.z}}{Je{\.z}}{2013}]%
        {jez2013}
\bibfield{author}{\bibinfo{person}{A. Je{\.z}}.}
  \bibinfo{year}{2013}\natexlab{}.
\newblock \showarticletitle{Recompression: a simple and powerful technique for
  word equations}. In \bibinfo{booktitle}{{\em Proc. {STACS} 2013}} {\em
  (\bibinfo{series}{LIPIcs})}, Vol.~\bibinfo{volume}{20}.
  \bibinfo{pages}{233--244}.
\newblock


\bibitem[\protect\citeauthoryear{Je{\.z}}{Je{\.z}}{2017}]%
        {jez2017}
\bibfield{author}{\bibinfo{person}{A. Je{\.z}}.}
  \bibinfo{year}{2017}\natexlab{}.
\newblock \showarticletitle{Word Equations in Nondeterministic Linear Space}.
  In \bibinfo{booktitle}{{\em Proc. {ICALP} 2017}} {\em
  (\bibinfo{series}{LIPIcs})}, Vol.~\bibinfo{volume}{80}.
  \bibinfo{pages}{95:1--95:13}.
\newblock


\bibitem[\protect\citeauthoryear{Jiang, Salomaa, Salomaa, and Yu}{Jiang
  et~al\mbox{.}}{1995}]%
        {jiang}
\bibfield{author}{\bibinfo{person}{T. Jiang}, \bibinfo{person}{A. Salomaa},
  \bibinfo{person}{K. Salomaa}, {and} \bibinfo{person}{S. Yu}.}
  \bibinfo{year}{1995}\natexlab{}.
\newblock \showarticletitle{Decision Problems for Patterns}.
\newblock \bibinfo{journal}{{\em J. Comput. Syst. Sci.\/}}
  \bibinfo{volume}{50}, \bibinfo{number}{1} (\bibinfo{year}{1995}),
  \bibinfo{pages}{53--63}.
\newblock


\bibitem[\protect\citeauthoryear{Karhum{\"a}ki, Mignosi, and
  Plandowski}{Karhum{\"a}ki et~al\mbox{.}}{2000}]%
        {karhumaki2000}
\bibfield{author}{\bibinfo{person}{J. Karhum{\"a}ki}, \bibinfo{person}{F.
  Mignosi}, {and} \bibinfo{person}{W. Plandowski}.}
  \bibinfo{year}{2000}\natexlab{}.
\newblock \showarticletitle{The expressibility of languages and relations by
  word equations}.
\newblock \bibinfo{journal}{{\em Journal of the ACM (JACM)\/}}
  \bibinfo{volume}{47}, \bibinfo{number}{3} (\bibinfo{year}{2000}),
  \bibinfo{pages}{483--505}.
\newblock


\bibitem[\protect\citeauthoryear{Kiezun, Ganesh, Guo, Hooimeijer, and
  Ernst}{Kiezun et~al\mbox{.}}{2009}]%
        {kiezun2009}
\bibfield{author}{\bibinfo{person}{A. Kiezun}, \bibinfo{person}{V. Ganesh},
  \bibinfo{person}{P.~J. Guo}, \bibinfo{person}{P. Hooimeijer}, {and}
  \bibinfo{person}{M.~D. Ernst}.} \bibinfo{year}{2009}\natexlab{}.
\newblock \showarticletitle{HAMPI: a solver for string constraints}. In
  \bibinfo{booktitle}{{\em Proc. ISSTA 2009}}. ACM, \bibinfo{pages}{105--116}.
\newblock


\bibitem[\protect\citeauthoryear{Lin and Barcel{\'o}}{Lin and
  Barcel{\'o}}{2016}]%
        {lin2016}
\bibfield{author}{\bibinfo{person}{A.~W. Lin} {and} \bibinfo{person}{P.
  Barcel{\'o}}.} \bibinfo{year}{2016}\natexlab{}.
\newblock \showarticletitle{String solving with word equations and transducers:
  towards a logic for analysing mutation XSS}. In \bibinfo{booktitle}{{\em ACM
  SIGPLAN Notices}}, Vol.~\bibinfo{volume}{51}. ACM, \bibinfo{pages}{123--136}.
\newblock


\bibitem[\protect\citeauthoryear{Lothaire}{Lothaire}{1983}]%
        {lothaire}
\bibfield{author}{\bibinfo{person}{M. Lothaire}.}
  \bibinfo{year}{1983}\natexlab{}.
\newblock \bibinfo{booktitle}{{\em Combinatorics on Words}}.
\newblock \bibinfo{publisher}{Addison-Wesley}.
\newblock


\bibitem[\protect\citeauthoryear{Makanin}{Makanin}{1977}]%
        {makanin1977}
\bibfield{author}{\bibinfo{person}{G.~S. Makanin}.}
  \bibinfo{year}{1977}\natexlab{}.
\newblock \showarticletitle{The problem of solvability of equations in a free
  semigroup}.
\newblock \bibinfo{journal}{{\em Sbornik: Mathematics\/}} \bibinfo{volume}{32},
  \bibinfo{number}{2} (\bibinfo{year}{1977}), \bibinfo{pages}{129--198}.
\newblock


\bibitem[\protect\citeauthoryear{Manea, Nowotka, and Schmid}{Manea
  et~al\mbox{.}}{2016}]%
        {DLT2016}
\bibfield{author}{\bibinfo{person}{F. Manea}, \bibinfo{person}{D. Nowotka},
  {and} \bibinfo{person}{M.~L. Schmid}.} \bibinfo{year}{2016}\natexlab{}.
\newblock \showarticletitle{On the Solvability Problem for Restricted Classes
  of Word Equations}. In \bibinfo{booktitle}{{\em Proc. {DLT} 2016}} {\em
  (\bibinfo{series}{LNCS})}, Vol.~\bibinfo{volume}{9840}.
  \bibinfo{pages}{306--318}.
\newblock


\bibitem[\protect\citeauthoryear{Matiyasevich}{Matiyasevich}{1968}]%
        {matiyasevich1968}
\bibfield{author}{\bibinfo{person}{Y.~V. Matiyasevich}.}
  \bibinfo{year}{1968}\natexlab{}.
\newblock \showarticletitle{A connection between systems of words-and-lengths
  equations and Hilbert's tenth problem}.
\newblock \bibinfo{journal}{{\em Zapiski Nauchnykh Seminarov POMI\/}}
  \bibinfo{volume}{8} (\bibinfo{year}{1968}), \bibinfo{pages}{132--144}.
\newblock


\bibitem[\protect\citeauthoryear{Morvan}{Morvan}{2000}]%
        {morvan2000}
\bibfield{author}{\bibinfo{person}{C. Morvan}.}
  \bibinfo{year}{2000}\natexlab{}.
\newblock \showarticletitle{On rational graphs}. In \bibinfo{booktitle}{{\em
  Proc. FoSSaCS 2000}} {\em (\bibinfo{series}{LNCS})},
  Vol.~\bibinfo{volume}{1784}. \bibinfo{pages}{252--266}.
\newblock


\bibitem[\protect\citeauthoryear{Nowotka and Saarela}{Nowotka and
  Saarela}{2016}]%
        {SaarelaDLT}
\bibfield{author}{\bibinfo{person}{D. Nowotka} {and} \bibinfo{person}{A.
  Saarela}.} \bibinfo{year}{2016}\natexlab{}.
\newblock \showarticletitle{One-Unknown Word Equations and Three-Unknown
  Constant-Free Word Equations}. In \bibinfo{booktitle}{{\em Proc. {DLT} 2016}}
  {\em (\bibinfo{series}{LNCS})}, Vol.~\bibinfo{volume}{9840}.
  \bibinfo{pages}{332--343}.
\newblock


\bibitem[\protect\citeauthoryear{Papadimitriou}{Papadimitriou}{1981}]%
        {papadimitriou1981}
\bibfield{author}{\bibinfo{person}{C.~H. Papadimitriou}.}
  \bibinfo{year}{1981}\natexlab{}.
\newblock \showarticletitle{On the Complexity of Integer Programming}.
\newblock \bibinfo{journal}{{\em Journal of the ACM (JACM)\/}}
  \bibinfo{volume}{28}, \bibinfo{number}{4} (\bibinfo{year}{1981}),
  \bibinfo{pages}{765--768}.
\newblock


\bibitem[\protect\citeauthoryear{Plandowski}{Plandowski}{1999}]%
        {plandowski1999}
\bibfield{author}{\bibinfo{person}{W. Plandowski}.}
  \bibinfo{year}{1999}\natexlab{}.
\newblock \showarticletitle{Satisfiability of word equations with constants is
  in PSPACE}. In \bibinfo{booktitle}{{\em Proc. FOCS 1999}}. IEEE,
  \bibinfo{pages}{495--500}.
\newblock


\bibitem[\protect\citeauthoryear{Quine}{Quine}{1946}]%
        {quine1946}
\bibfield{author}{\bibinfo{person}{W.~V. Quine}.}
  \bibinfo{year}{1946}\natexlab{}.
\newblock \showarticletitle{Concatenation as a basis for arithmetic}.
\newblock \bibinfo{journal}{{\em J. Symb. Log.\/}} \bibinfo{volume}{11},
  \bibinfo{number}{4} (\bibinfo{year}{1946}), \bibinfo{pages}{105--114}.
\newblock


\bibitem[\protect\citeauthoryear{Trinh, Chu, and Jaffar}{Trinh
  et~al\mbox{.}}{2016}]%
        {trinh2016}
\bibfield{author}{\bibinfo{person}{M.-T. Trinh}, \bibinfo{person}{D.-H. Chu},
  {and} \bibinfo{person}{J. Jaffar}.} \bibinfo{year}{2016}\natexlab{}.
\newblock \showarticletitle{Progressive Reasoning over Recursively-Defined
  Strings}. In \bibinfo{booktitle}{{\em Proc. CAV 2016}} {\em
  (\bibinfo{series}{LNCS})}, Vol.~\bibinfo{volume}{9779}.
  \bibinfo{pages}{218--240}.
\newblock


\bibitem[\protect\citeauthoryear{Vazenin and Rozenblat}{Vazenin and
  Rozenblat}{1983}]%
        {russians}
\bibfield{author}{\bibinfo{person}{J.~M. Vazenin} {and} \bibinfo{person}{B.~V.
  Rozenblat}.} \bibinfo{year}{1983}\natexlab{}.
\newblock \showarticletitle{Decidability of the positive theory of a free
  countably generated semigroup}.
\newblock \bibinfo{journal}{{\em Math. USSR Sb.\/}}  \bibinfo{volume}{44.1}
  (\bibinfo{year}{1983}), \bibinfo{pages}{109--116}.
\newblock


\bibitem[\protect\citeauthoryear{Yu, Alkhalaf, and Bultan}{Yu
  et~al\mbox{.}}{2010}]%
        {yu2010}
\bibfield{author}{\bibinfo{person}{F. Yu}, \bibinfo{person}{M. Alkhalaf}, {and}
  \bibinfo{person}{T. Bultan}.} \bibinfo{year}{2010}\natexlab{}.
\newblock \showarticletitle{STRANGER: An Automata-based String Analysis Tool
  for PHP}. In \bibinfo{booktitle}{{\em Proc. TACAS 2010}} {\em
  (\bibinfo{series}{LNCS})}, Vol.~\bibinfo{volume}{6015}.
\newblock


\end{thebibliography}

\end{document}